\documentclass[pra,aps,showpacs,onecolumn,twoside,superscriptaddress]{revtex4}
\usepackage{amsmath,amsfonts,amssymb,caption,color,epsfig,graphics,graphicx,hyperref,latexsym,mathrsfs,revsymb,theorem,url,verbatim,epstopdf}

\hypersetup{colorlinks,linkcolor={blue},citecolor={red},urlcolor={blue}}

\newtheorem{definition}{Definition}
\newtheorem{proposition}[definition]{Proposition}
\newtheorem{lemma}[definition]{Lemma}

\newtheorem{theorem}[definition]{Theorem}
\newtheorem{corollary}[definition]{Corollary}
\newtheorem{conjecture}[definition]{Conjecture}

\newtheorem{remark}[definition]{Remark}
\newtheorem{example}[definition]{Example}
\newtheorem{question}[definition]{Question}

\def\bcj{\begin{conjecture}}
\def\ecj{\end{conjecture}}
\def\bcr{\begin{corollary}}
\def\ecr{\end{corollary}}
\def\bd{\begin{definition}}
\def\ed{\end{definition}}
\def\bea{\begin{eqnarray}}
\def\eea{\end{eqnarray}}
\def\bem{\begin{enumerate}}
\def\eem{\end{enumerate}}
\def\bex{\begin{example}}
\def\eex{\end{example}}
\def\bim{\begin{itemize}}
\def\eim{\end{itemize}}
\def\bl{\begin{lemma}}
\def\el{\end{lemma}}
\def\bma{\begin{bmatrix}}
\def\ema{\end{bmatrix}}
\def\bpf{\begin{proof}}
\def\epf{\end{proof}}
\def\bpp{\begin{proposition}}
\def\epp{\end{proposition}}
\def\bqu{\begin{question}}
\def\equ{\end{question}}
\def\br{\begin{remark}}
\def\er{\end{remark}}
\def\bt{\begin{theorem}}
\def\et{\end{theorem}}

\def\squareforqed{\hbox{\rlap{$\sqcap$}$\sqcup$}}
\def\qed{\ifmmode\squareforqed\else{\unskip\nobreak\hfil
\penalty50\hskip1em\null\nobreak\hfil\squareforqed
\parfillskip=0pt\finalhyphendemerits=0\endgraf}\fi}
\def\endenv{\ifmmode\;\else{\unskip\nobreak\hfil
\penalty50\hskip1em\null\nobreak\hfil\;
\parfillskip=0pt\finalhyphendemerits=0\endgraf}\fi}
\newenvironment{proof}{\noindent \textbf{{Proof.~} }}{\qed}
\def\Dbar{\leavevmode\lower.6ex\hbox to 0pt
{\hskip-.23ex\accent"16\hss}D}
\makeatletter
\def\url@leostyle{%
  \@ifundefined{selectfont}{\def\UrlFont{\sf}}{\def\UrlFont{\small\ttfamily}}}
\makeatother
\def\bcj{\begin{conjecture}}
\def\ecj{\end{conjecture}}
\def\bcr{\begin{corollary}}
\def\ecr{\end{corollary}}
\def\bd{\begin{definition}}
\def\ed{\end{definition}}
\def\bea{\begin{eqnarray}}
\def\eea{\end{eqnarray}}
\def\bem{\begin{enumerate}}
\def\eem{\end{enumerate}}
\def\bex{\begin{example}}
\def\eex{\end{example}}
\def\bim{\begin{itemize}}
\def\eim{\end{itemize}}
\def\bl{\begin{lemma}}
\def\el{\end{lemma}}
\def\bpf{\begin{proof}}
\def\epf{\end{proof}}
\def\bpp{\begin{proposition}}
\def\epp{\end{proposition}}
\def\bqu{\begin{question}}
\def\equ{\end{question}}
\def\br{\begin{remark}}
\def\er{\end{remark}}
\def\bt{\begin{theorem}}
\def\et{\end{theorem}}

\def\btb{\begin{tabular}}
\def\etb{\end{tabular}}

\newcommand{\nc}{\newcommand}

\def\a{\alpha}

\def\m{\mu}

\def\s{\sigma}

 \nc{\bbA}{\mathbb{A}} \nc{\bbB}{\mathbb{B}} \nc{\bbC}{\mathbb{C}}
 \nc{\bbD}{\mathbb{D}} \nc{\bbE}{\mathbb{E}} \nc{\bbF}{\mathbb{F}}
 \nc{\bbG}{\mathbb{G}} \nc{\bbH}{\mathbb{H}} \nc{\bbI}{\mathbb{I}}
 \nc{\bbJ}{\mathbb{J}} \nc{\bbK}{\mathbb{K}} \nc{\bbL}{\mathbb{L}}
 \nc{\bbM}{\mathbb{M}} \nc{\bbN}{\mathbb{N}} \nc{\bbO}{\mathbb{O}}
 \nc{\bbP}{\mathbb{P}} \nc{\bbQ}{\mathbb{Q}} \nc{\bbR}{\mathbb{R}}
 \nc{\bbS}{\mathbb{S}} \nc{\bbT}{\mathbb{T}} \nc{\bbU}{\mathbb{U}}
 \nc{\bbV}{\mathbb{V}} \nc{\bbW}{\mathbb{W}} \nc{\bbX}{\mathbb{X}}
 \nc{\bbZ}{\mathbb{Z}}

 \nc{\bA}{{\bf A}} \nc{\bB}{{\bf B}} \nc{\bC}{{\bf C}}
 \nc{\bD}{{\bf D}} \nc{\bE}{{\bf E}} \nc{\bF}{{\bf F}}
 \nc{\bG}{{\bf G}} \nc{\bH}{{\bf H}} \nc{\bI}{{\bf I}}
 \nc{\bJ}{{\bf J}} \nc{\bK}{{\bf K}} \nc{\bL}{{\bf L}}
 \nc{\bM}{{\bf M}} \nc{\bN}{{\bf N}} \nc{\bO}{{\bf O}}
 \nc{\bP}{{\bf P}} \nc{\bQ}{{\bf Q}} \nc{\bR}{{\bf R}}
 \nc{\bS}{{\bf S}} \nc{\bT}{{\bf T}} \nc{\bU}{{\bf U}}
 \nc{\bV}{{\bf V}} \nc{\bW}{{\bf W}} \nc{\bX}{{\bf X}}
 \nc{\bZ}{{\bf Z}}

\nc{\cA}{{\cal A}} \nc{\cB}{{\cal B}} \nc{\cC}{{\cal C}}
\nc{\cD}{{\cal D}} \nc{\cE}{{\cal E}} \nc{\cF}{{\cal F}}
\nc{\cG}{{\cal G}} \nc{\cH}{{\cal H}} \nc{\cI}{{\cal I}}
\nc{\cJ}{{\cal J}} \nc{\cK}{{\cal K}} \nc{\cL}{{\cal L}}
\nc{\cM}{{\cal M}} \nc{\cN}{{\cal N}} \nc{\cO}{{\cal O}}
\nc{\cP}{{\cal P}} \nc{\cQ}{{\cal Q}} \nc{\cR}{{\cal R}}
\nc{\cS}{{\cal S}} \nc{\cT}{{\cal T}} \nc{\cU}{{\cal U}}
\nc{\cV}{{\cal V}} \nc{\cW}{{\cal W}} \nc{\cX}{{\cal X}}
\nc{\cZ}{{\cal Z}}

\nc{\hA}{{\hat{A}}} \nc{\hB}{{\hat{B}}} \nc{\hC}{{\hat{C}}}
\nc{\hD}{{\hat{D}}} \nc{\hE}{{\hat{E}}} \nc{\hF}{{\hat{F}}}
\nc{\hG}{{\hat{G}}} \nc{\hH}{{\hat{H}}} \nc{\hI}{{\hat{I}}}
\nc{\hJ}{{\hat{J}}} \nc{\hK}{{\hat{K}}} \nc{\hL}{{\hat{L}}}
\nc{\hM}{{\hat{M}}} \nc{\hN}{{\hat{N}}} \nc{\hO}{{\hat{O}}}
\nc{\hP}{{\hat{P}}} \nc{\hR}{{\hat{R}}} \nc{\hS}{{\hat{S}}}
\nc{\hT}{{\hat{T}}} \nc{\hU}{{\hat{U}}} \nc{\hV}{{\hat{V}}}
\nc{\hW}{{\hat{W}}} \nc{\hX}{{\hat{X}}} \nc{\hZ}{{\hat{Z}}}

\nc{\hn}{{\hat{n}}}

\def\max{\mathop{\rm max}}

\newcommand{\ket}[1]{|#1\rangle}

\newcommand{\braket}[2]{\langle#1|#2\rangle}

\newcommand{\tbc}{\red{TO BE CONTINUED...}}

\newcommand{\red}{\textcolor{red}}
\def\Dbar{\leavevmode\lower.6ex\hbox to 0pt
{\hskip-.23ex\accent"16\hss}D}
\begin{document}

\title{The construction and local distinguishability of multiqubit unextendible product bases}

\date{\today}

\pacs{03.65.Ud, 03.67.Mn}

\author{Yize Sun}\email[]{sunyize@buaa.edu.cn}
\affiliation{School of Mathematical Sciences, Beihang University, Beijing 100191, China}

\author{Lin Chen}\email[]{linchen@buaa.edu.cn (corresponding author)}
\affiliation{School of Mathematical Sciences, Beihang University, Beijing 100191, China}
\affiliation{International Research Institute for Multidisciplinary Science, Beihang University, Beijing 100191, China}

\begin{abstract}
An important problem in quantum information is to construct multiqubit unextendible product bases (UPBs). By using the unextendible orthogonal matrices, we construct a $7$-qubit UPB of size 11. 
It solves an open problem in [Quantum Information Processing 19:185 (2020)].  %We apply our results in graph theory. 
Next, we graph-theoretically show that the UPB is locally indistinguishable in the bipartite systems of two qubits and five qubits, respectively. It turns out that the UPB corresponds to a complete graph with 11  vertices constructed by three sorts of nonisomorphic graphs. Taking the graphs as product vectors, we show that they are in three different orbits up to local unitary equivalence. Moreover, we also present the number of sorts of nonisomorphic graphs of complete graphs of some known UPBs and their orbits. %$q$-qubit UPBs with size $q+1$ for odd $q$ and some special UPBs for even $q=4,6,8$ with size $6,8,11$. Our results show the substantial connection between multiqubit UPBs and complete graphs. 
\end{abstract}

%\Large

\maketitle

%\tableofcontents

\section{Introduction}

Unextendible product bases (UPBs) are one of the most versatile objects in the study of the multiqubit positive-partial-transpose (PPT)
entangled states \cite{Bennett1998Unextendible,1998Quantum,1999Unextendible}, indecomposible positive maps \cite{2012A}, local discrimination \cite{AL01,dms03,Chen2013The} and genuinely entangled spaces \cite{br01,PhysRevLett.99.250405,Chen2018The}.
Multiqubit UPBs help reliably build quantum circuits and cryptography in experiments \cite{Dicarlo2010Preparation}. Theoretically, 
%\cite{Johnston_2014}. 
 multiqubit UPBs can be used to construct tight Bell inequalities without quantum violation \cite{2012Tight}. The studies of UPBs begin with the size of them \cite{AL01}. Since then, the existence of  UPBs with distinct sizes has received extensive attentions \cite{Fen06,2019Constructing,Wang2014Unextendible%cite-key
 }. It is shown that  all 4-qubit UPBs are constructed \cite{Johnston_2014}. Recently, 7-qubit UPBs of size 10 %$10\times7$ unextendible orthogonal matrix (UOM)
  have been constructed \cite{ISI:000534346500003} and an open problem is whether the 7-qubit UPBs of size
11 exist. %i.e. the $11\times7$ UOM corresponds to a 7-qubit UPB of size 10. 
We solve the problem by constructing such UPBs. This is the first motivation of this paper.  

The local distinguishability of a given set of states is an important problem connected with the local
operations and classical communications (LOCC) and many schemes have been widely considered \cite{2002Nonlocality,2002Local,2001Local}.  %many schemes for locally distinguishing between a set of pure and mixed quantum states have been widely considered \cite{2002Nonlocality,2002Local,2001Local}.  
It has been applied in various quantum information tasks, such as data hiding \cite{Matthews2009Distinguishability,2002Hiding},  quantum secret sharing \cite{Markham2012Graph} and protocols of teleportation \cite{1993Teleporting}.  Local orthogonal quantum states can always be distinguishable \cite{2004Distinguishability}. 
Great efforts have been devoted to the set of actions on the multipartite system under LOCC 
only \cite{AL01}. For example, any set of orthogonal product states in $2\times N$ is distinguishable by LOCC  \cite{1999Unextendible} and two orthogonal pure states can be distinguishable \cite{Jonathan2000Local}. 
A UPB is not distinguishable by LOCC, though they contain no entanglement \cite{PhysRevLett.82.5385}. 
As far as we know, there are few schemes of  distinguishing  UPBs by collective   measurements. We shall locally indistinguish the 7-qubit UPB of size 11 using a collective measurement over the bipartite system of any two and five qubits, respectively.
%For a concrete 7-qubit UPB of size 11, if it exists then we consider whether it is locally distinguishable. 
This is the second motivation of this paper. 
 
Graph theory may quite properly be regarded as an area of applied mathematics \cite{2006Graph} and has many applications in quantum information \cite{2015Graph,Duan2014No,Gravier2012Quantum,Majewski2014On,2009QUANTUM}. %, such as chemistry \cite{1985Applications}, brain network \cite{van2010Comparing} and quantum physics \cite{Fen06}. 
For example, it has been applied in qubit information systems of extremal black branes \cite{2015Graph}, no-signalling assisted zero-error capacity of quantum channels \cite{Duan2014No} and quantum secret sharing \cite{Gravier2012Quantum}. Ref. \cite{Chen2013The} has shown the minimum size of UPBs by using graph theory techniques. It provides an operational tool when dealing with unextendible product
bases, particularly in the qubit case. 
We regard product vectors of a UPB as vertices in a graph. The orthogonality of product vectors corresponds to the edge between vertices. If any two product vectors have exactly one orthogonal pair of local vectors then there is exactly one edge between any two vertices. It implies that the UPB corresponds to a complete graph in graph theory. We may investigate the properties of UPBs by using graph theory. This is the third motivation of this paper. 

%Complete graphs in graph theory mean that there is exactly one edge between any two vertices.

%Some results have shown the connection between complete graphs and UPBs. For example 
 
% \red{introduce some background of graph theory and UPBs....}

In this paper, we construct a 7-qubit UPB of size 11  in Theorem \ref{thm:11times7} by investigating some general properties of it in Lemma \ref{le:submatrix} and \ref{le:pj}. %i.e., the $11\times7$ UOM corresponds to a 7-qubit UPB of size 11. 
The UPB turns out to be locally indistinguishable in Theorem \ref{thm:distinguish}. We also show a graph-theoretic proof of the same theorem in Lemma \ref{fig:graph} by using the orthogonality of each vertex with the remaining vertices in Fig. \ref{fig:uom117dl}. Next, we obtain that the 7-qubit UPB of size 11 in Theorem \ref{thm:11times7} corresponds to a complete graph with eleven vertices. By investigating the orthogonality of every graph in Fig. \ref{fig:uom117}, we present that the product vectors corresponding to them are in three different orbits up to local unitary (LU) equivalence in Lemma \ref{le:lu}. %\red{By investigating the orthogonality of each vertex with remaining vertices in Fig. \ref{fig:uom117dl}, we obtain the same conclusion as Theorem \ref{thm:distinguish}  using graph theoretic. }
 Moreover, we investigate the connections between complete graphs and $q$-qubit UPBs of size $q+1$, where $q$ is odd. In Lemma 	\ref{le:eqodd}, we obtain that for any odd $q (\geq3)$, one of $q$-qubit UPBs of size $q+1$ corresponds to a complete graph with $(q+1)$ vertices. It is constructed with exactly one sort of nonisomorphic graphs. Similarly, by assuming every graph for constructing the complete graph as a product vector, the product vectors are in the same orbit up to LU equivalence and permutation of systems.  
 For some even cases, in Theorem \ref{le:eqsort}, we show that 4,6,8-qubit UPBs of size $6,8,11$ correspond to complete graphs, respectively. They can be constructed with three sorts of nonisomorphic graphs. We further show that the product vectors corresponding to the graphs in Fig. \ref{fig:uom64de}, \ref{fig:uom86de}, and \ref{fig:uom118de} are in three different orbits, respectively.
 
%Furthermore, UPBs can be used to construct subspaces of small dimension that are nonlocally distinguishable. Subspaces with dimension 4 can be constructed from three-qubit UPBs and have bases distinguishable by separable operations \cite{2010Locally}. Recently, two protocols have  been constructed to locally distinguish a particular UPBs in $\mathbb{C}^5\otimes\mathbb{C}^5$ by using different entanglement resource \cite{2020Locally}.
%Our results may share a new viewpoint to these topic. 
%\red{ Since the 7-qubit UPB of size 11 is distinguishable in $2\times 2^6$, we 	system A may be teleported to system B so that B can collectively measure system A and B, by doing it, A and B consume one ebit. Similarly, we can construct the scheme of realizing measurement on system 2-5, which costs 1+4=5 ebits.}
 
The rest of this paper is organized as follows. In section \ref{sect:iii}, we introduce the notations and facts on UPBs. We construct a concrete 7-qubit UPB of size 11 in Theorem \ref{thm:11times7} and  we apply it in Theorem \ref{thm:distinguish}. In section \ref{sect:complete}, we investigate connections between UPBs and complete graphs in Lemma \ref{le:eqodd} and Theorem \ref{le:eqsort}.  
Finally, we conclude in section \ref{con:117}.

\section{An example of $11\times7$ UOMs and local distinguishability}
\label{sect:iii}

The $p$-qubit pure quantum state is represented by a unit vector $\ket{v}\in(\mathcal{H}^2)^{\otimes p}$, which is called a {product state} if it can be decomposed in the following form:
\begin{eqnarray}
	\ket{v}=\ket{v_1}\otimes\ket{v_2}\otimes\cdots\otimes\ket{v_p}, \quad \ket{v_j}\in\mathbb{C}^2,\forall j.
\end{eqnarray}
The standard basis of $\mathcal{H}^2$ is $\{\ket{0},\ket{1}\}$ and we use $\{\ket{a},\ket{a'}\}$, $\{\ket{b},\ket{b'}\}$,$\{\ket{c},\ket{c'}\},\cdots$ to denote the orthonormal bases of $\mathcal{H}^2$ different from $\{\ket{0},\ket{1}\}$ and from each other. We also will sometimes find it useful to omit the tensor product symbol when discussing multi-qubit states. For example, we use $\ket{0,1,1,0}$ as a shorthand way to write $\ket{0}\otimes\ket{1}\otimes\ket{1}\otimes\ket{0}$.

A $p$-qubit UPB is a set $\mathcal{S}\subseteq(\mathcal{H}^2)^{\otimes p}$ satisfying the following three
properties:

(i) every $\ket{v}\in\mathcal{S}$ is a product state;

(ii) $\braket{v}{w}=0$ for all $\ket{v}\neq\ket{w}\in\mathcal{S}$; 

(iii) all product states $\ket{z}\notin\mathcal{S}$, then there exists $\ket{v}\in\mathcal{S}$ such that $\braket{v}{z}\neq0$.

So a UPB is a set of orthogonal product states such that there is no product state orthogonal to every member of the set. %Then we introduce a fact from \cite{Johnston_2014} as follows.

%\begin{lemma}
%	Let $p\in\mathbb{N}$. If there exist $p$-qubit UPBs $\mathcal{S}_1$ and $\mathcal{S}_2$ with $|\mathcal{S}_1|=s_1$ and $|\mathcal{S}_2|=s_2$. Then there exists a $(p+1)$-qubit UPB $\mathcal{S}$ with $|\mathcal{S}|=s_1+s_2$.
%\end{lemma}

%Orthogonality graph is a tool to visualize UPBs and simplify proofs. For a set of product states $\mathcal{S}=\{\ket{v_1},\ket{v_2},\cdots,\ket{v_s}\}\in(\mathbb{C}^2)^{\otimes p}$ with $|\mathcal{S}|=s$. The \textit{orthogonality graph} of  $\mathcal{S}$ is the graph on $s$ vertices $V:=\{v_1,\cdots,v_s\}$ such that there is an edge $(v_i,v_j)$ of color $l$ if and only if $\ket{v_i}$ and $\ket{v_j}$ are orthogonal to each other on qubit $l$. Because the members of a UPB are mutually orthogonal is equivalent to
%requiring that every edge is present on at least one qubit in its orthogonality graph, the orthogonality graph is an edge coloring of the complete graph. We know that two qubit UPBs to be equivalent if they have the  same orthogonality graphs up to permuting the qubits and relabeling the vertices \cite{Johnston_2014}.

One can verify that the maximum size of a $p$-qubit UPB is $2^p$ since the standard basis forms a UPB. It is well-known that there are no nontrivial $p$-qubit UPBs when $p\leq2$ \cite{Bennett1998Unextendible}.  %So the first case of interest is when $p=3$. The only one \textit{nontrivial} three-qubit UPB  is $\mathcal{C}=\{\ket{0,0,0},\ket{1,+,-},\ket{-,1,+},\ket{+,-,1}\}$, where $\ket{+}=\frac{\ket{0}+\ket{1}}{\sqrt{2}}$ and $\ket{-}=\frac{\ket{0}-\ket{1}}{\sqrt{2}}$. 
A nontrivial $p$-qubit UPB means that one whose size is strictly less than $2^p$.  We refer to the %\textit{
size of a UPB as the number of states in the UPB.
%Based on the nontrivial three-qubit UPB $\mathcal{C}$, there exists a nontrivial $p$-qubit UPB of size $2^p-4$ for all $p\geq3$. 
We review the minimum size of $p$-qubit UPBs found in \cite{Johnston2013The}.

%Ref. \cite{Johnston2013The} shows that the minimum size of $p$-qubit UPBs in $(\mathcal{H}^2)^{\otimes p}$. We introduce the results as follows.
\begin{lemma}
	\label{le:functionp}
	Denote a function $f(p)$ as the minimum size of $p$-qubit UPB, then  	
	
	(i) if $p$ is odd then $f(p) = p + 1$;
	
	(ii) if $p= 4$ or $p\equiv2$ (mod 4) then $f(p) = p + 2$;
	
	(iii) if $p = 8$ then $f(p) = p + 3$;
	
	(iv) otherwise, $f(p) = p + 4$.
\end{lemma}
%\begin{eqnarray}
% f(x)=\left\{
%\begin{aligned}
%p+1&&\text{if $p$ is odd}\\
%p+2&&\text{if $p=4$ or $p\equiv2(\mod4)$} \\
%p+3&&\text{if $p=8$}\\
%p+4&&\text{otherwise},
%\end{aligned}
%\right.
%\end{eqnarray}

We recall the concept of unextendible orthogonal matrix (UOM) \cite{ISI:000534346500003}. 
Denote product vectors of a $n$-qubit UPB of size $m$ as row vectors of a $m\times n$ matrix. The matrix is known as the UOM of the UPB, and its rows are orthogonal. For the orthogonal qubit vectors $\ket{a}$ and $\ket{a'}$, denote them as the vector variables $a$ and $a'$ in UOMs. For example, the three-qubit UPB $\{\ket{0,0,0},\ket{1,b,c},\ket{a,1,c'},\ket{a',b',1}\}$ can be expressed as the $4\times3$ UOM
\begin{eqnarray}
	\bma
	0&0&0\\
	1&b&c\\
	a&1&c'\\
	a'&b'&1
	\ema.
\end{eqnarray}

In the following, %we investigate some properties of $11\times7$ UOMs in Lemma \ref{le:submatrix} and \ref{le:pj}. W
we present a concrete $11\times7$ UOM in Theorem \ref{thm:11times7}. Since an $11\times7$ UOM corresponds to a 7-qubit UPB of size 11, its product vectors are orthogonal.

\begin{theorem}
	\label{thm:11times7}
	There exists an $11\times7$ UOM
	\begin{eqnarray}
	\label{eq:11times7uom}
\bma
a_{1,1}&a_{1,2}&a_{1,3}&a_{1,4}&a_{1,5}&a_{1,6}&a_{1,7}\\
a_{1,1}&a_{2,2}&a_{2,3}&a_{2,4}&a_{1,5}'&a_{2,6}&a_{2,7}\\	a_{1,1}'&a_{2,2}&a_{1,3}&a_{3,4}&a_{3,5}&a_{3,6}&a_{3,7}\\
a_{1,1}'&a_{1,2}&a_{2,3}&a_{4,4}&a_{4,5}&a_{4,6}&a_{3,7}'\\
a_{4,1}&a_{1,2}'&a_{5,3}&a_{2,4}'&a_{5,5}&a_{3,6}'&a_{5,7}\\
a_{4,1}&a_{1,2}'&a_{1,3}&a_{3,4}'&a_{6,5}&a_{2,6}'&a_{5,7}'\\
a_{4,1}'&a_{7,2}&a_{2,3}'&a_{7,4}&a_{3,5}'&a_{1,6}'&a_{7,7}\\
a_{4,1}'&a_{7,2}'&a_{2,3}'&a_{3,4}'&a_{6,5}&a_{8,6}&a_{1,7}'\\
a_{9,1}&a_{7,2}'&a_{1,3}'&a_{4,4}'&a_{5,5}'&a_{8,6}'&a_{2,7}'\\
a_{9,1}&a_{7,2}&a_{1,3}'&a_{7,4}'&a_{4,5}'&a_{2,6}'&a_{5,7}'\\
a_{9,1}'&a_{2,2}'&a_{5,3}'&a_{1,4}'&a_{6,5}'&a_{4,6}'&a_{7,7}'
\ema.
\end{eqnarray}
\end{theorem}

We show the proof of above theorem in Appendix \ref{eq:pftheom6}. %For Theorem \ref{thm:distinguish}, we also give a graph-theoretic proof of it in the next section. In addition, 
In addition, we investigate the $n$-qubit UPBs of size $m$, where $n=7$ and $8$. By known facts \cite{ ISI:000534346500003} and Theorem \ref{thm:11times7}, we conclude the findings that $7$-qubit UPBs of size $m$ exist when  $m=8,10-122,124,128$. From \cite{Chen_2018} and  \cite{ ISI:000534346500003}, we obtain that 8-qubit UPBs of size $m$ exist when $m=11-250,252,256$. So we have exhausted all sizes of $n=7$ and 8. %For 9-qubit UPBs of size $m$, it is unknown  whether there exist $9$-qubit UPBs of size $13,14,15,17-21$ \cite{Chen_2018}. So the next open problem is to determine the existence 9-qubit UPBs of size 13. % and Appendix $A$ of \cite{Chen_2018} have shown $m=11,\cdots,250,252,256\in\Theta_8$ and $m=251,253,254,255\notin\Theta_8$. 

%\section{Application}
%\label{sec:distinguish}
%In this section, 
%we consider whether $7$-qubit UPB of size $11$  in section \ref{sect:iii} can construct a set of orthogonal bases such that it is locally distinguishable. %It is also an open problem in \cite{Shi_2020}. 
%A set of orthogonal states is locally distinguished if there exists a sequence of LOCC can distinguish the states. A measurement performed to distinguish a set of orthogonal states is called an orthogonality-preserving measurement if the states remain orthogonal after the measurement. Local distinguishability sufficiently ensures local reducibility. A set of orthogonal quantum states is called a locally reducible set if it is possible to distinguish one or more states from the set by orthogonality-preserving local measurement. 

Since the $11\times7$ UOM  in Eq. \eqref{eq:11times7uom} corresponds to a 7-qubit UPB of size 11 $\mathcal{K}$, we consider whether it can construct a set of orthogonal bases such that it is locally distinguishable. Because  any set of orthogonal product states in $2\times N$ is distinguishable \cite{1999Unextendible}, we assume that the product vectors of $\mathcal{K}$ are bipartite states of systems $\otimes_{j\in\{1,\cdots,7\}\backslash\{6\}} C_j$ and $C_6$. For the convenience of readers, we show it in Fig. \ref{fig:reducible1}.

	\begin{figure}[htb]
	\includegraphics[scale=0.9,angle=0]{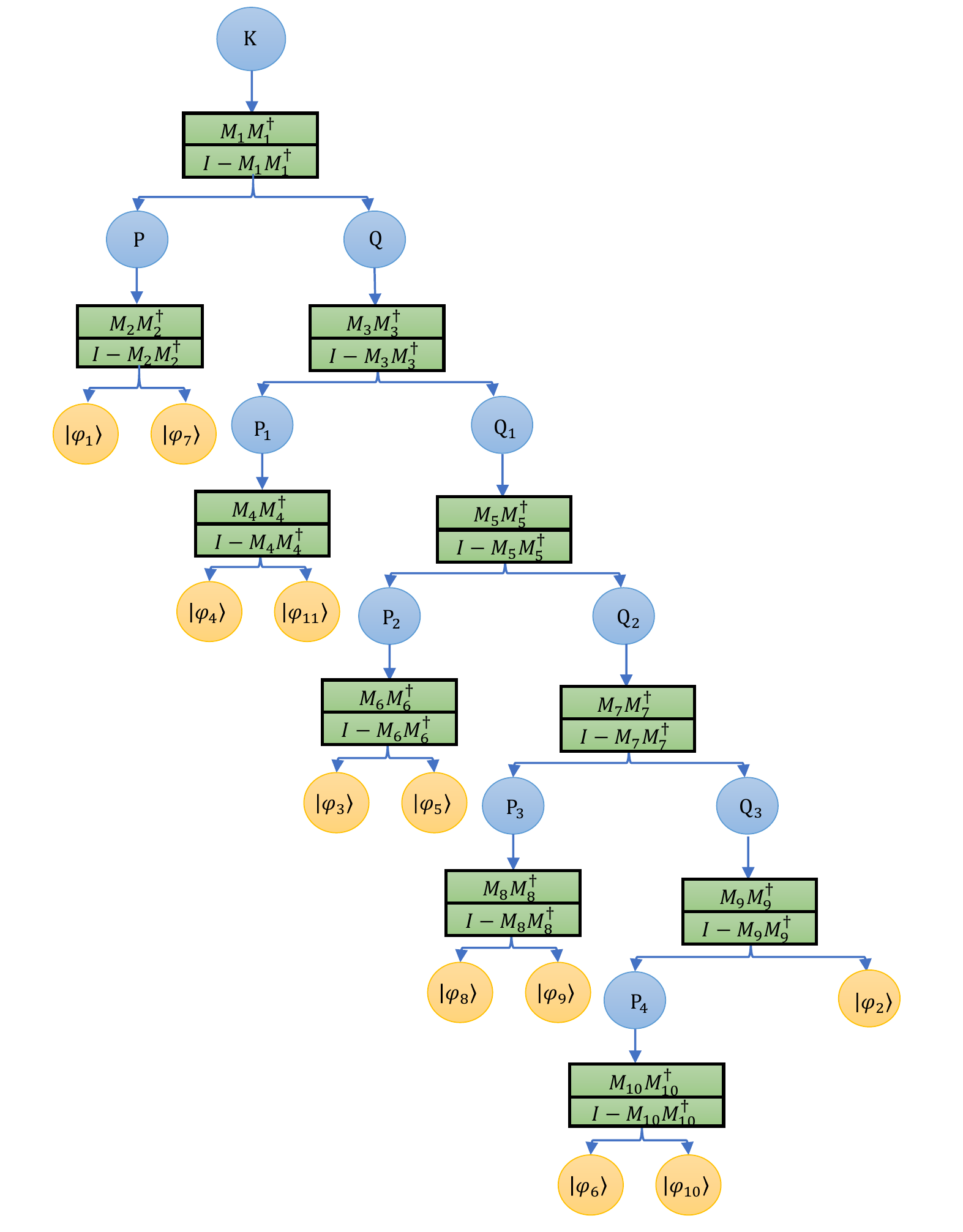}
	\caption{This is the description of steps. Consider every element $\ket{\varphi_i}, i\in\{1,\cdots,11\}$ of the set $\mathcal{K}$ as a bipartite state of systems $\otimes_{j\in\{1,\cdots,7\}\backslash\{6\}} C_j$ and $C_6$. The set $\mathcal{K}$ is locally distinguishable by using measurements $\{M_jM_j^\dag,I-M_jM_j^\dag\}$, where $j\in\{1,\cdots,10\}$.}
	\label{fig:reducible1}
\end{figure}

\begin{figure}[htb]
	\includegraphics[scale=0.9,angle=0]{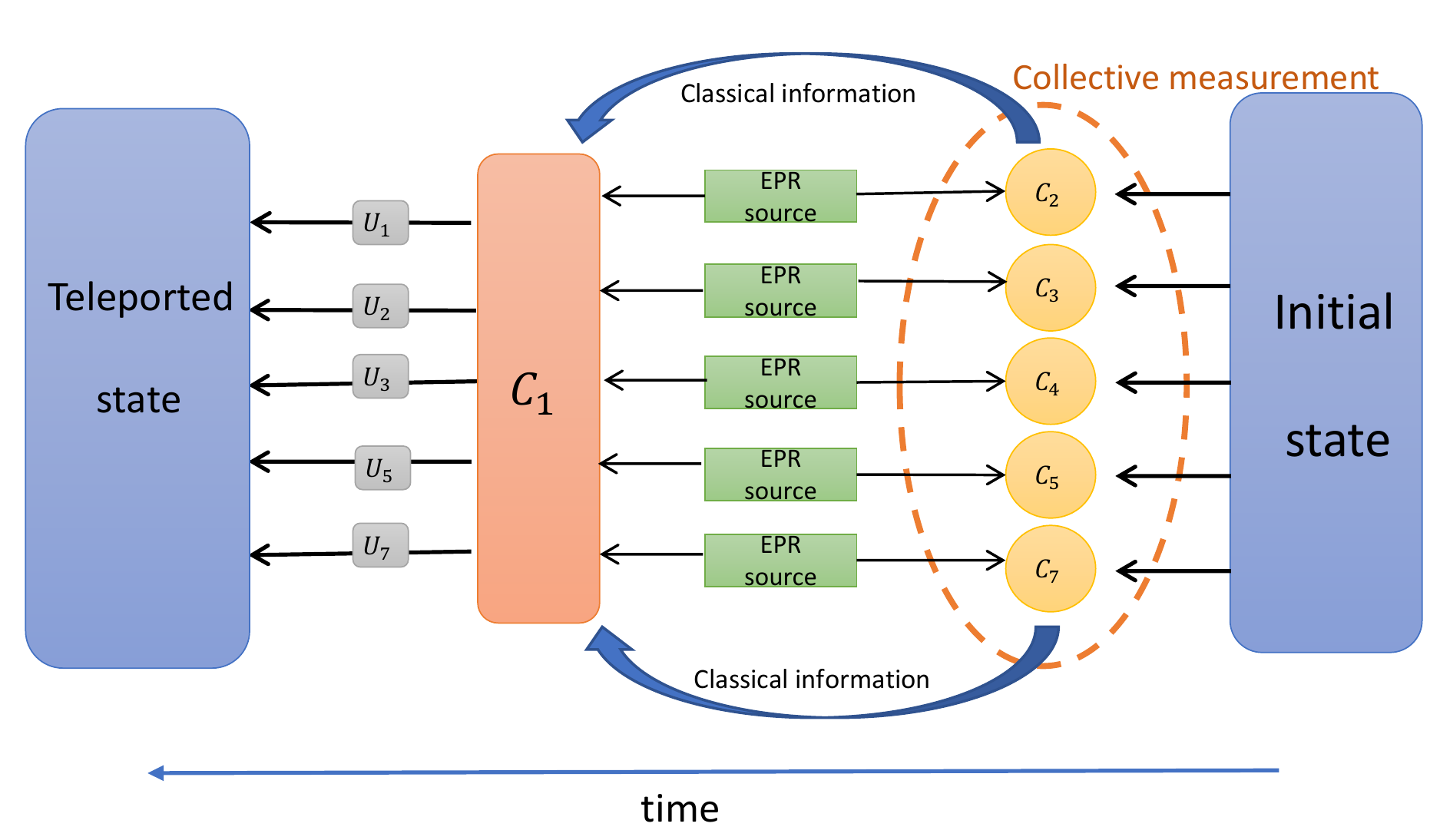}
	\caption{Quantum teleportation is a way of utilizing the
		entangled EPR pair in order to send initial state to Bob, with only a small overhead of classical
		communication. %This process cost entanglement  1 ebit. 
		Systems $C_2,C_3,C_4,C_5,C_7$ can teleport their particles to $C_1$ such that $C_1$ owns the six systems $C_1,\cdots,C_5,C_7$. Then we can use collective measurement on them, it costs entanglement 5*1 ebit =5 ebits.} 
	\label{fig:teleport}
\end{figure}

Physically, the collective measurement of system $C_1,\cdots, C_5,C_7$ can be realized using quantum teleportation in Fig. \ref{fig:teleport}. Suppose $C_1$ and $C_2$ share a two-qubit maximally entangled state. Using quantum teleportation, $C_2$ may teleport his particle to $C_1$. Similarly, using the two-qubit  maximally entangled state between $C_1$ and $C_j$ for $j=3,4,5,7$, $C_j$ may teleport its particle to $C_1$. Then $C_1$ owns the six systems $C_1,\cdots,C_5,C_7$. So it can perform measurement on them. This collective measurement cost entanglement 5 ebits. We further consider whether $\mathcal{K}$ is locally distinguishable  in the bipartite systems of two qubits and five qubits, in the following theorem.

% section \ref{sect:iii} 

%We consider whether the product vectors of the  $11\times7$ UOM in Eq. \eqref{eq:11times7uom} is locally distinguishable in the following theorem.
%A set of orthogonal states is strongly nonlocal if it is locally irreducible in every bipartition \cite{Shi_2020}. \tbc

%\begin{conjecture}
%	If a UPB corresponds to a complete graph, then it is a locally irreducible set in every bipartition.
%\end{conjecture}
\begin{theorem}
	\label{thm:distinguish}
	Suppose the product states of the $11\times7$ UOM in Eq. \eqref{eq:11times7uom} %corresponding to a complete graph 
	are on systems $C_1,\cdots,C_7$. Regard them as bipartite states of systems $\otimes_{j\in\{m,n\}}C_j$  and $\otimes_{j\in\{1,\cdots,7\}\backslash\{m,n\}} C_j$. Then they are locally indistinguishable.
\end{theorem}
 
We show the proof in Appendix \ref{pf:theorem7}. We also give a graph-theoretic proof of Theorem \ref{thm:distinguish} in the next section. Naturally, from Theorem \ref{thm:distinguish}, we may further consider whether the product states of the $11\times7$ UOM in \eqref{eq:11times7uom} are locally distinguishable in the bipartite systems of three qubits and four qubits. Since a strongly nonlocal orthogonal set cannot be locally distinguishable in every bipartition \cite{Shi_2020}, an open problem is whether the product states of the $11\times7$ UOM in \eqref{eq:11times7uom} have strong nonlocality in multipartite system. Because any set of orthogonal product states in $2\times N$ is distinguishable by LOCC  \cite{1999Unextendible}, we exclude it when we investigate the open problem.

%have strong nonlocalityThe result shows that the UOM has the strong nonlocality in multipartite system, i.e. ,%\red{ Ref. \cite{2017From} has shown that the relation between  UPBs and genuinely entangled subspaces. It has provided ways of constructing non–orthogonal UPBs leading to subspaces containing solely states being genuinely entangled in their orthocomplement. We may further consider to construct genuinely entangled states by using the states in Theorem \ref{thm:distinguish}. }

%We show the proofs of above two theorem in Appendix \ref{eq:pftheom6} and \ref{pf:theorem7}. For Theorem \ref{thm:distinguish}, we also give a graph-theoretic proof of it in the next section. In addition, we investigate the size $m$ of $n$-qubit UPBs, where $n=7,8,9$.   By known facts \cite{ ISI:000534346500003} and Theorem \ref{thm:11times7}, we conclude the findings that $7$-qubit UPBs of size $m$ exist when  $m=8,10-122,124,128$. From \cite{Chen_2018} and  \cite{ ISI:000534346500003}, we obtain that 8-qubit UPBs of size $m$ exist when $m=11-250,252,256$. For 9-qubit UPBs of size $m$, it is unknown  whether there exist $9$-qubit UPBs of size $13,14,15,17-21$ \cite{Chen_2018}. So the next open problem is to determine the existence 9-qubit UPBs of size 13. % and Appendix $A$ of \cite{Chen_2018} have shown $m=11,\cdots,250,252,256\in\Theta_8$ and $m=251,253,254,255\notin\Theta_8$. 
\section{Complete graphs of UOMs }
\label{sect:complete}

In this section, we show a graph-theoretic proof of Theorem \ref{thm:distinguish}. We investigate the properties of  the complete graph in Fig. \ref{fig:uom117} corresponding to the 7-qubit UPB of size 11 in \eqref{eq:11times7uom}. By investigating the orthogonality of every graphs in Fig. \ref{fig:uom117}, if we assume that  every graph corresponds to a product vector, then  we further consider the number of orbits up to LU equivalence. Next, for some other known UPBs in Lemma \ref{le:functionp}, we investigate connections between them and the  complete graphs corresponding to them in Lemma \ref{le:eqodd} and Theorem \ref{le:eqsort}.

 We recall the definition of complete graphs. 
A graph with exactly one edge between two vertices is called a {complete graph}. 
%we construct $p$-qubit UOMs of $m$ state by using two $q_1,q_2$-qubit UOMs of $m_1,m_2$ states, where $q_1,q_2\leq q, m_1,m_2\leq m$. \tbc 
We know that the number of edges of a complete graph with $n$ vertices is $\frac{n(n-1)}{2}$. For a UOM, let one row correspond to one vertex. If two rows have orthogonal pairs of vector variables, then there exist edges between the two vertices. The number of orthogonal pairs of any two rows equals to the number of edges between any two vertices. So for a UOM, if any two rows have exactly one orthogonal pair, then we obtain that there exists only one edge between any two vertices. Hence, the $11\times7$ UOM in Theorem \ref{thm:11times7} whose any two rows has one orthogonal pair corresponds to a complete graph with eleven vertices. We denote $V_i$ as the vertex corresponding to the $i$-th row, $1\leq i\leq11$. Then  we present the graphs corresponding to cases of orthogonal pairs of column $1-7$ in Fig. \ref{fig:uom117} and  the complete graph with eleven vertices in Fig. \ref{fig:uom117de}.

%\red{Note: the $11\times7$ UOM in Theorem \ref{thm:11times7} has three inequivalent graphs in Fig. \ref{fig:uom117}. }

\begin{figure}[htb]
	\includegraphics[scale=0.5,angle=0]{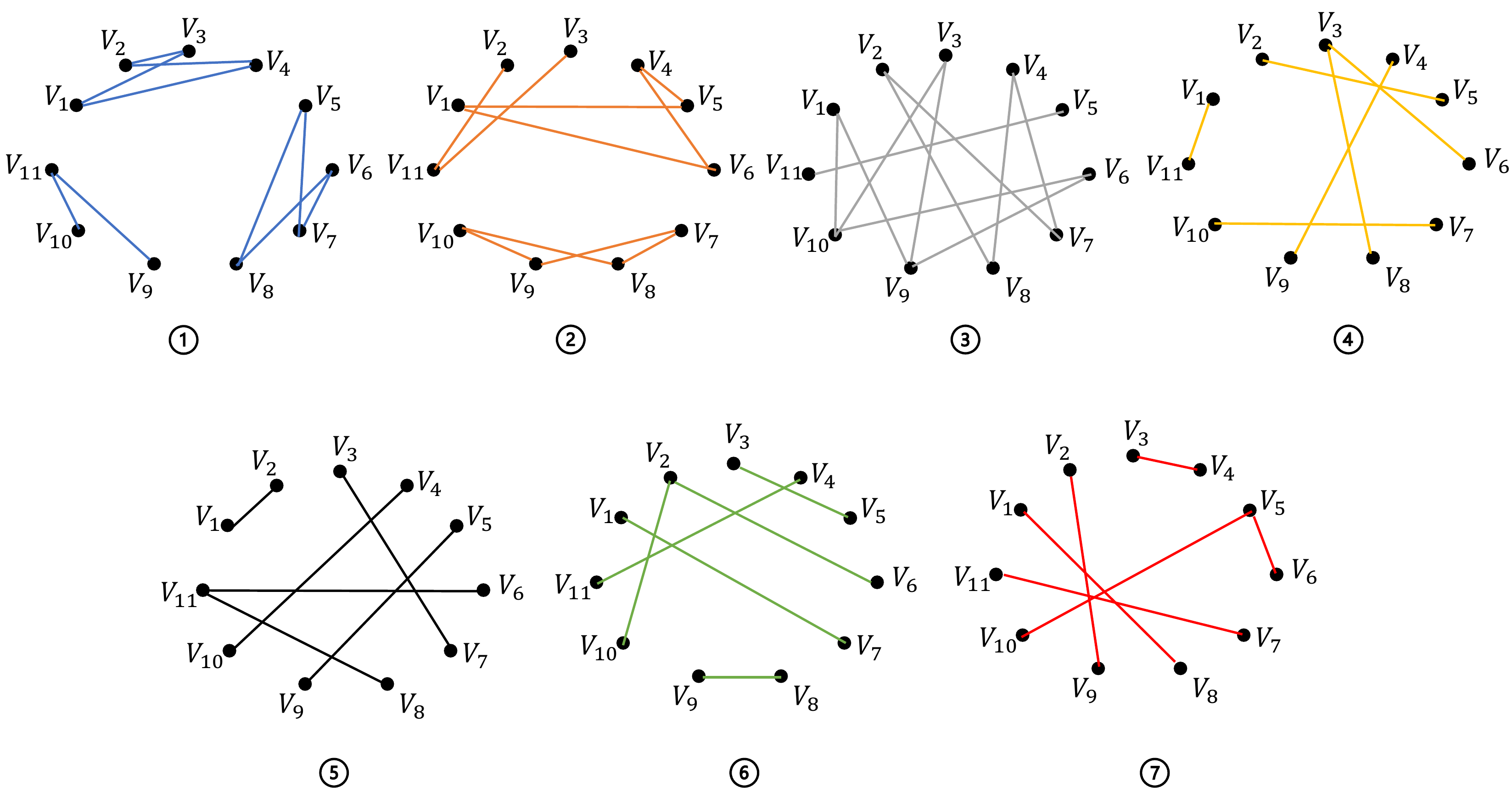}
	\caption{Graph \textcircled{1} - \textcircled{7} correspond to the cases of orthogonal pairs in column $1-7$ of the $11\times7$ UOM of Theorem \ref{thm:11times7}, respectively.} 
	\label{fig:uom117}
\end{figure}

\begin{figure}[htb]
	\includegraphics[scale=0.5,angle=0]{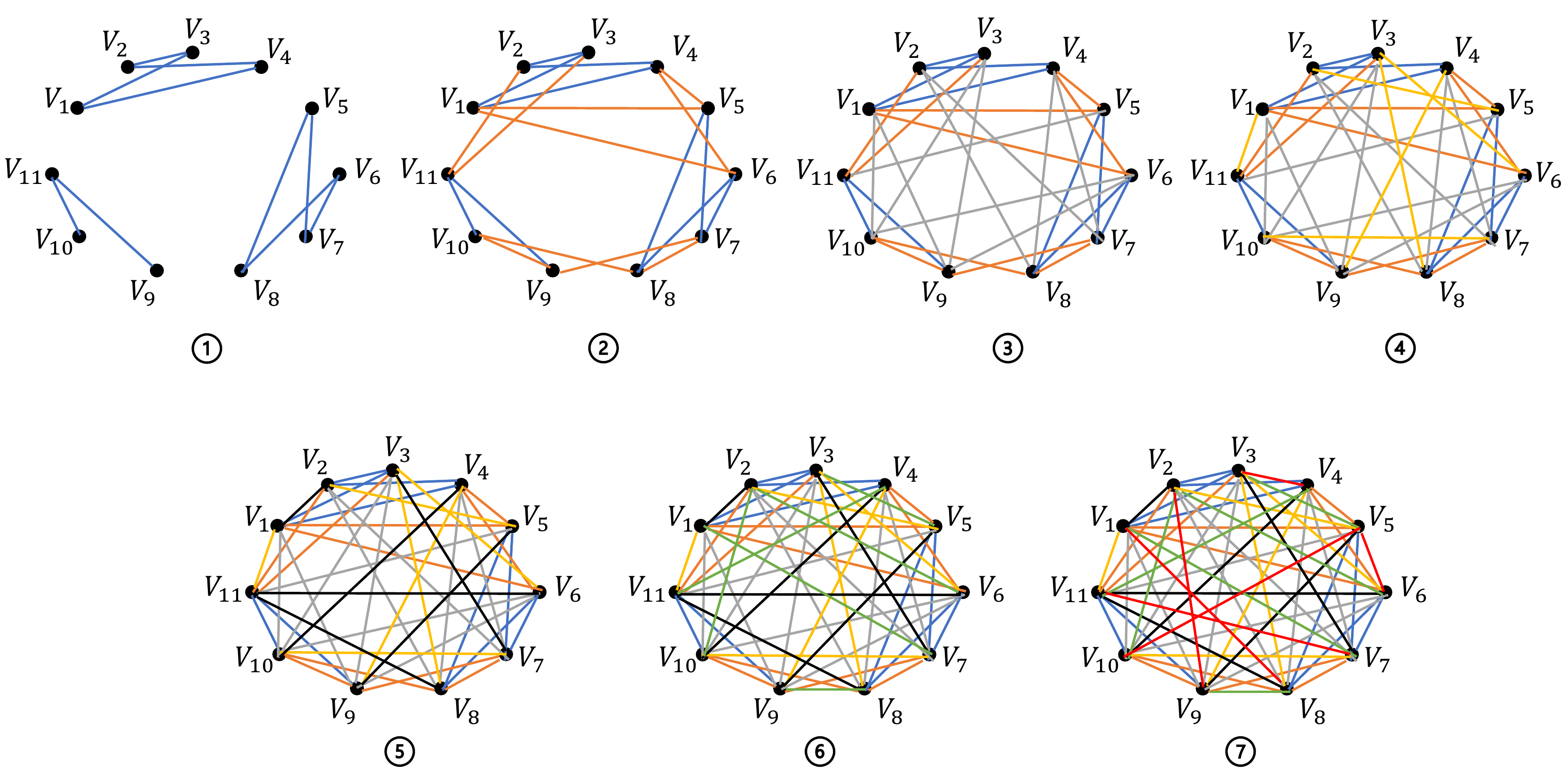}
	\caption{Graph \textcircled{1} - \textcircled{7} show the cases of orthogonal pairs in the first column $1-7$. The $11\times7$ UOM of Theorem \ref{thm:11times7} corresponds to the complete graph in \textcircled{3}. 
	} 
	\label{fig:uom117de}
\end{figure}
%We introduce a fact from \cite{Johnston_2014} as follows.

Since the $11\times7$ UOM corresponds to a 7-qubit UPB $\mathcal{K}$ of size 11 in \eqref{eq:11times7uom}, we obtain that $\mathcal{K}$ corresponds to the complete graph in Fig. \ref{fig:uom117de}. It implies that every product vector of it corresponds to a vertex of the graph. The representation of the graph in Fig. \ref{fig:uom117} can help verify whether the product vectors of $\mathcal{K}$ are locally distinguishable as follows. 

\begin{lemma}
	\label{fig:graph}
	Theorem \ref{thm:distinguish} can be also obtained by using the graph-theoretic technique.
\end{lemma}

\begin{figure}[htb]
	\includegraphics[scale=0.7,angle=0]{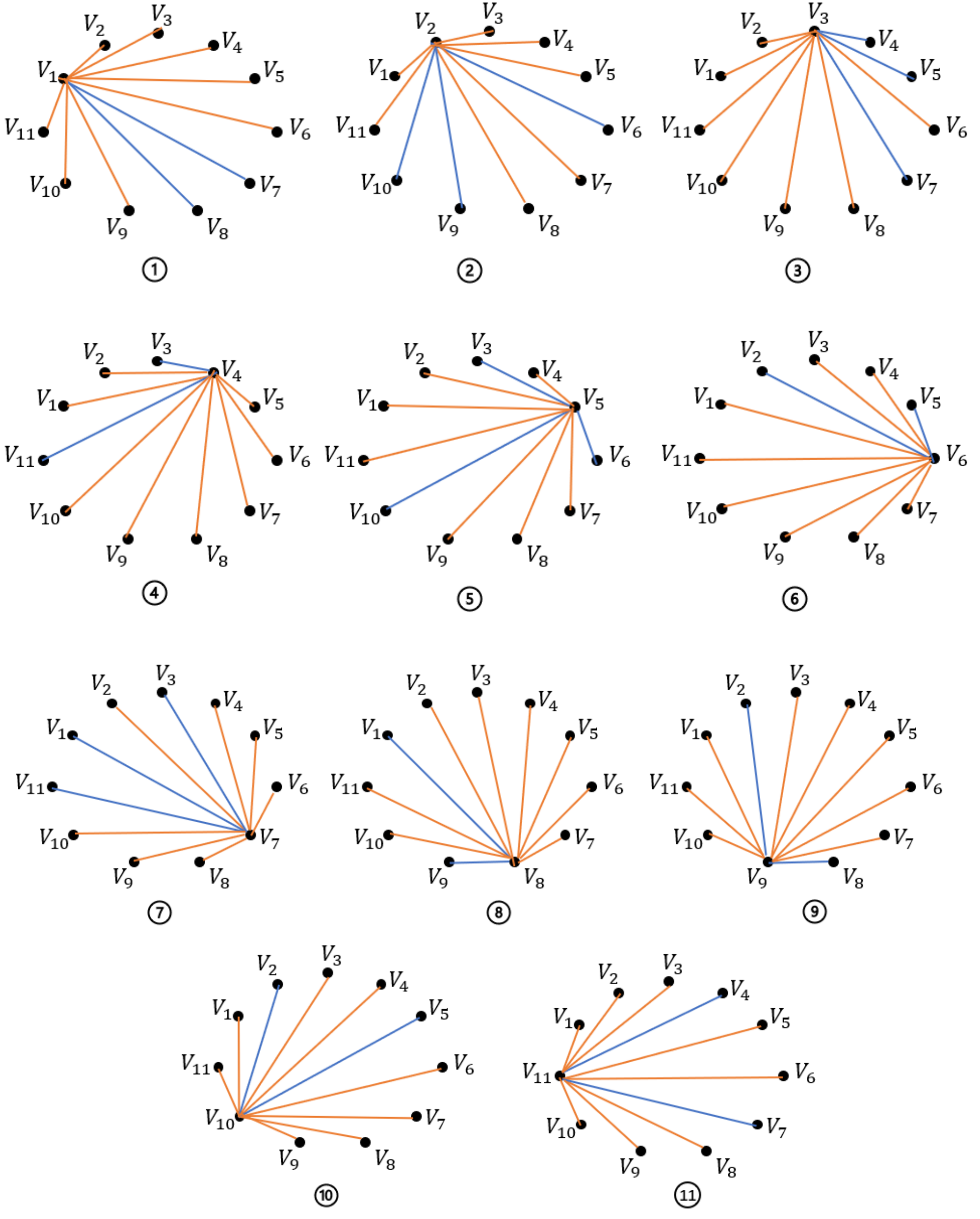}
	\caption{Every $\ket{\varphi_i}$ corresponds to vertex $V_i$. The figure shows that every $\ket{\varphi_i}$ is orthogonal to the other on two systems by different colors. If any two product vectors $\ket{\varphi_i}$ and $\ket{\varphi_k}$ are orthogonal on system $\otimes_{j\in\{6,7\}}C_j$, then we may assume that there exists a blue edge between $V_i$ and $V_k$. If any two product vectors $\ket{\varphi_i}$ and $\ket{\varphi_k}$ are orthogonal on system $\otimes_{j\in\{1,\cdots,5\}} C_j$, then we may assume that there exists a yellow edge between $V_i$ and $V_k$, where $i,k\in\{1,\cdots,11\}$.} 
	\label{fig:uom117dl}
\end{figure}

\begin{proof}
We may assume that every product vector $\ket{\varphi_i}$ is on systems $C_1,\cdots, C_7$, where $i=1,\cdots,11$. We consider every $\ket{\varphi_i}$ of $\mathcal{K}$ as a bipartite states of systems $\otimes_{j\in\{m,n\}}C_j$  and $\otimes_{j\in\{1,\cdots,7\}\backslash\{m,n\}} C_j$. By investigating orthogonality of any two product vectors on systems $\otimes_{j\in\{m,n\}}C_j$ or $\otimes_{j\in\{1,\cdots,7\}\backslash\{m,n\}} C_j$, we present it by using different colors. We start with an example when $m=6$ and $n=7$, respectively. Then we show it in Fig. \ref{fig:uom117dl}. %For the convenience of description, we denote systems  $\otimes_{j\in\{6,7\}}C_j$ and $ \otimes_{j\in\{1,\cdots,5\}} C_j$ as $K_1$ and $K_2$. 

We may assume that 
the product vectors of $\mathcal{K}$ are locally distinguishable on system $\otimes_{j\in\{6,7\}}C_j$. Then we obtain that $\mathcal{K}$ can be splitted into two disjoint subsets $P$ and $Q$ and the elements of $P$ are orthogonal to that of $Q$. In Fig. \ref{fig:uom117dl}, we obtain that $\ket{\varphi_1}$ is orthogonal to $\ket{\varphi_i}, i\in\{2,\cdots,11\}\backslash \{7,8\}$ on system $\otimes_{j\in\{1,\cdots,5\}} C_j$. Because any two product vectors have exactly one orthogonal pair, we obtain that  $\ket{\varphi_1}$ is not orthogonal to $\ket{\varphi_i}, i\in\{2,\cdots,11\}\backslash \{7,8\}$ on system $\otimes_{j\in\{6,7\}}C_j$. So we obtain that  $\ket{\varphi_1}$ and  $\ket{\varphi_i}, i\in\{2,\cdots,11\}\backslash \{7,8\}$ are in the same set. Without loss of generality, we may assume that they are in set $P$. Since $\ket{\varphi_1}$ is orthogonal to $\ket{\varphi_7}$ and $\ket{\varphi_8}$, the remaining elements of $P$ are orthogonal to  $\ket{\varphi_7}$ or $\ket{\varphi_8}$. However, 
from the $\textcircled{2}$ in Fig. \ref{fig:uom117dl}, we obtain that it is not orthogonal to $\ket{\varphi_7}$ and $\ket{\varphi_8}$. Hence, we obtain that $\mathcal{K}$ are locally indistinguishable on system $\otimes_{j\in\{6,7\}}C_j$. 

Next, we assume that the product vectors of $\mathcal{K}$ are locally distinguishable on system $\otimes_{j\in\{1,\cdots,5\}} C_j$. Then we obtain that $\mathcal{K}$ can be splitted into two disjoint subsets $P$ and $Q$ on system $\otimes_{j\in\{1,\cdots,5\}} C_j$. Without loss of generality, we may assume that $\ket{\varphi_1}$ is in $P$. From $\textcircled{1}$ of Fig. \ref{fig:uom117dl}, we obtain that $\ket{\varphi_1}$ is orthogonal to $\ket{\varphi_7}$ and $\ket{\varphi_8}$ on system $\otimes_{j\in\{6,7\}}C_j$. Since any two product vectors have exactly one orthogonal pair, we obtain that $\ket{\varphi_1}$ is not orthogonal to $\ket{\varphi_7}$ and $\ket{\varphi_8}$. So they are in set $P$. From $\textcircled{7}$ and $\textcircled{8}$ of Fig. \ref{fig:uom117dl}, we obtain that $\ket{\varphi_3}, \ket{\varphi_9}, \ket{\varphi_{11}}$ are in set $P$. Similarly, from $\textcircled{3}, \textcircled{9}$ and $\textcircled{11}$ of Fig. \ref{fig:uom117dl}, we obtain that $\ket{\varphi_4}, \ket{\varphi_5}, \ket{\varphi_{2}}$ are in set $P$. From $\textcircled{2}$ of Fig. \ref{fig:uom117dl}, we obtain that $\ket{\varphi_6}, \ket{\varphi_{10}}$ are in set $P$. So we obtain that all product vectors of $\mathcal{K}$ are in set $P$. It implies that there exists no product vector in set $Q$. It violates  the fact that the product vectors of $\mathcal{K}$ are locally distinguishable on system $\otimes_{j\in\{1,\cdots,5\}} C_j$. 

Similarly, by using the same methods,  when $m,n\in\{1,\cdots,7\}$ and $m\neq n$, we obtain that the product vectors of $\mathcal{K}$ are locally indistinguishable on systems $\otimes_{j\in\{m,n\}}C_j$ or $\otimes_{j\in\{1,\cdots,7\}\backslash\{m,n\}} C_j$. Thus the result is the same as that in Theorem \ref{thm:distinguish} by using a graph-theoretic proof. 
\end{proof}

Furthermore, we investigate more results about some known UOMs in Lemma \ref{le:functionp} and their corresponding complete graphs. 
In graph theory, suppose $G(V,E)$ and $G_1(V_1,E_1)$ are two graphs, where $V,V_1$ are two sets of vertices and $E,E_1$ are sets of edges in graph $G, G_1$, respectively. If there exists a bijective $M:V\rightarrow V_1$ such that for all $X,Y\in V$, the edge $XY\in E$ is equivalent to the edge $M(X)M(Y)\in E_1$, then $G$ and $G_1$ are {isomorphic}. Then in Fig.  \ref{fig:uom117}, \textcircled{1} and \textcircled{2} are two {isomorphic} graphs. Similarly,  \textcircled{4}, \textcircled{5}, \textcircled{6},  \textcircled{7} are isomorphic. 
%So we obtain that there are three non{isomorphic} graphs. 
Thus the $11\times7$ UOM in Theorem \ref{thm:11times7} is constructed by three sorts of nonisomorphic graphs. We assume that every graph in Fig. \ref{fig:uom117} corresponds to a product vector. In the following lemma, we show the number of orbits of them up to LU equivalence. 

\begin{lemma}
	\label{le:lu}
	Suppose every graph in Fig. \ref{fig:uom117} corresponds to  a product vector in $(\mathcal{H}^2)^{\otimes11}$. Then they are in three different orbits up to LU equivalence.
\end{lemma}
\begin{proof}
From Fig. \ref{fig:uom117}, it implies that the orthogonality of every column in Eq. \eqref{eq:11times7uom}. If we regard every column of Eq. \eqref{eq:11times7uom} as a product vector $\ket{\phi_i}\in(\mathcal{H}^2)^{\otimes {11}},i=1,\cdots,7$, then we may assume that 
\begin{eqnarray}
	\label{eq:11product}
\ket{\phi_1}:=	\ket{a_{1,1},a_{1,1},a_{1,1}',a_{1,1}',a_{4,1},a_{4,1},a_{4,1}',a_{4,1}',a_{9,1},a_{9,1},a_{9,1}'},\nonumber
&&
\\
\ket{\phi_2}:=\ket{a_{1,2},a_{2,2},a_{1,2},a_{2,2},a_{1,2}',a_{1,2}',a_{7,2},a_{7,2}',a_{7,2}',a_{7,2},a_{2,2}'},
\nonumber
&&
\\
\ket{\phi_3}:=\ket{a_{1,3},a_{2,3},a_{1,3},a_{2,3},a_{5,3},a_{1,3},a_{2,3}',a_{2,3}',a_{1,3}',a_{1,3}',a_{5,3}'},
\nonumber
&&
\\
\ket{\phi_4}:=\ket{a_{1,4},a_{2,4},a_{3,4},a_{4,4},a_{2,4}',a_{3,4}',a_{7,4},a_{3,4}',a_{4,4}',a_{7,4}',a_{1,4}'},
\nonumber
&&
\\
\ket{\phi_5}:=\ket{a_{1,5},a_{1,5}',a_{3,5},a_{4,5},a_{5,5},a_{6,5},a_{3,5}',a_{6,5},a_{5,5}',a_{4,5}',a_{6,5}'},
\nonumber
&&
\\
\ket{\phi_6}:=\ket{a_{1,6},a_{2,6},a_{3,6},a_{4,6},a_{3,6}',a_{2,6}',a_{1,6}',a_{8,6},a_{8,6}',a_{2,6}',a_{4,6}'},
\nonumber
&&
\\
\ket{\phi_7}:=\ket{a_{1,7},a_{2,7},a_{3,7},a_{3,7}',a_{5,7},a_{5,7}',a_{7,7},a_{1,7}',a_{2,7}',a_{5,7}',a_{7,7}'}.
\end{eqnarray}

By using the permutation of systems on $\ket{\phi_2}$, we can rewrite $\ket{\phi_2}$ as
\begin{eqnarray}
	\label{eq:widetilde}
\ket{\widetilde{{\phi_2}}}=\ket{a_{1,2},a_{1,2},a_{1,2}',a_{1,2}',a_{7,2},a_{7,2},a_{7,2}',a_{7,2}',a_{2,2},a_{2,2},a_{2,2}'}.	
\end{eqnarray}  From the orthogonal pairs of $\ket{\phi_1}$ and $\ket{\widetilde{{\phi_2}}}$ in \eqref{eq:11product} and \eqref{eq:widetilde}, there exist $2\times2$ unitary matrices $U_1,U_2,U_3$ such that $U_1\ket{a_{1,1}}=\ket{a_{1,2}}, U_1\ket{a_{1,1}'}=\ket{a_{1,2}'}, U_2\ket{a_{4,1}}=\ket{a_{7,2}}, U_2\ket{a_{4,1}'}=\ket{a_{7,2}'}, U_3\ket{a_{9,1}}=\ket{a_{2,2}}, U_3\ket{a_{9,1}'}=\ket{a_{2,2}'}$.  We obtain that  there exists a product unitary operation $U$ such that $U\ket{\phi_1}=\ket{\phi_2}$, where $U=(U_1)^{\otimes4}\otimes (U_2)^{\otimes4}\otimes (U_3)^{\otimes3}$. From the definition of the LU equivalence, we obtain that $\ket{\phi_1}, \ket{\widetilde{\phi_2}}$ are LU equivalent and they are in the same orbit. Similarly, by using the  permutation of systems on $\ket{\phi_5},\ket{\phi_6},\ket{\phi_7}$, we may assume that they are
\begin{eqnarray}
	\ket{\widetilde{\phi_5}}=\ket{a_{1,5},a_{3,5},a_{6,5}',a_{5,5},a_{3,5}',a_{6,5},a_{4,5}',a_{6,5},a_{5,5}',a_{4,5},a_{1,5}'},
	&&
	\\
	\ket{\widetilde{\phi_6}}=\ket{a_{1,6},a_{3,6},a_{2,6},a_{4,6},a_{3,6}',a_{2,6}',a_{8,6},a_{2,6}',a_{4,6}',a_{8,6}',a_{1,6}'},
	&&
	\\
	\ket{\widetilde{\phi_7}}=\ket{a_{1,7},a_{2,7},a_{5,7},a_{3,7},a_{2,7}',a_{5,7}',a_{7,7},a_{5,7}',a_{3,7}',a_{7,7}',a_{1,7}'}.
\end{eqnarray}
Then we obtain that there exist three product unitary operations such that  $\ket{\widetilde{\phi_5}},\ket{\widetilde{\phi_6}},\ket{\widetilde{\phi_7}}$ are LU equivalent to $\ket{\widetilde{\phi_4}}$. So they are in the same orbit. We consider the same number of orthogonal pairs of these product vectors in \eqref{eq:11product}. 
By using LU equivalent, we obtain that $\ket{\phi_1},\ket{\phi_2}$ are LU equivalent. Similarly,   $\ket{\phi_4},\ket{\phi_5},\ket{\phi_6},\ket{\phi_7}$ are LU equivalent and  $\ket{\phi_3}$ is LU equivalent to itself. Hence, we obtain that the product vectors in \eqref{eq:11product} are in three different orbits up to LU equivalence and using permutation of systems.
\end{proof}

For some known UOMs in Lemma \ref{le:functionp}, we consider the number of sorts of {nonisomorphic} graphs for constructing the complete graphs corresponding to them in the following lemma. 
\begin{lemma}
	\label{le:eqodd}
For any odd $q\geq3$, one of $(q+1)\times q$ UOMs corresponds to a complete graph with $(q+1)$ vertices. It is constructed with exactly one sort of nonisomorphic graphs. 
\end{lemma}

The proof above lemma has been shown in Appendix \ref{le:lemma5}. By applying the above same methods, we obtain that all product vectors corresponding to the graphs for constructing complete graph are in the same orbit.  Next, we consider some cases when $q$ is even. From Lemma \ref{le:functionp} (ii) and (iii), we obtain that the minimum sizes of $4,6,8$-qubit UOMs are $6,8,11$, respectively. Then we consider the number of sorts of {nonisomorphic} graphs for constructing such UOMs as follows. 

\begin{theorem}
	\label{le:eqsort}
	There exist  $6\times4, 8\times6$ and $11\times8$ UOMs corresponding to complete graphs with $6, 8$ and $11$ vertices, respectively. Each of the three complete graphs can be constructed by three sorts of nonisomorphic graphs.
\end{theorem}

We show the proof of this theorem in Appendix \ref{eq:pf}. By using the same method in Lemma \ref{le:lu}, we obtain that the product vectors corresponding to the graphs in Fig. \ref{fig:uom64de}, \ref{fig:uom86de}, \ref{fig:uom118de} are in three different orbits. In addition, if a graph $g_1$ is a subgraph of another graph $g_2$, then we denote $g_1\subseteq g_2$. If graph $g_1$ and $g_2$ are isomorphic, then we denote $g_1\simeq g_2$. we use $\textcircled{1}_k, k=4,6,8$ to distinguish different \textcircled{1} in Fig. \ref{fig:uom64de}, \ref{fig:uom86de} and  \ref{fig:uom118de}. This method is also applicable to other graphs. 
In Fig. \ref{fig:uom64de}, \ref{fig:uom86de} and  \ref{fig:uom118de}, we obtain that $\textcircled{1}_4\subseteq\textcircled{1}_6\subseteq\textcircled{2}_8\simeq\textcircled{3}_8\simeq\textcircled{4}_8$, $\textcircled{2}_4\subseteq\textcircled{2}_6\simeq\textcircled{3}_6\subseteq\textcircled{1}_8$ and $\textcircled{3}_4\simeq\textcircled{4}_4\subseteq\textcircled{4}_6\simeq\textcircled{5}_6\simeq\textcircled{6}_6\subseteq\textcircled{5}_8\simeq\textcircled{6}_8\simeq\textcircled{7}_8\simeq\textcircled{8}_8$. So we obtain that graphs in Fig. \ref{fig:uom64de} are subgraphs of graphs in Fig. \ref{fig:uom86de} and graphs in Fig. \ref{fig:uom86de} are subgraphs of graphs in Fig. \ref{fig:uom118de}. Hence, we can construct $11\times8$ UOMs by $6\times4$ or $8\times6$ UOMs and construct $8\times6$ UOMs by $6\times4$ UOMs.

\section{Conclusions}
\label{con:117}
We have shown the existence of $7$-qubit UPBs of size $11$ and presented its concrete structure. Then we have shown that it is locally indistinguishable in the bipartite systems of two qubits and five qubits. We have applied our results in graph theory. We have shown that the 7-qubit UPB of size 11 corresponds to a complete graph with 11 vertices and the graph is constructed by three sorts of nonisomorphic graphs. Up to local unitary equivalence, we have shown that they are in three different orbits by taking the graphs as product vectors. 
 Moreover, for $q$-qubit UPBs of size $q+1$ for odd $q$, we have shown that the complete graphs corresponding to them are constructed by one sort of nonisomorphic graphs and they are in one orbit. 
For the minimum sizes of $4, 6, 8$-qubit UPBs, the complete graphs corresponding to them are constructed by three sorts of nonisomorphic graphs and they are in three different orbits. 

It is unknown  whether there exist $9$-qubit UPBs of size $13,14,15,17-21$ \cite{Chen_2018}. So one open problem is to determine the existence 9-qubit UPBs of size 13. Another problem is whether a tripartite state $\rho_{ABC}$ can be constructed by the 7-qubit UPB of size 11 of $\mathcal{H}^2\otimes \mathcal{H}^2\otimes(\mathcal{H}^{2})^{\otimes5}$ such that $\rho_{ABC}$ is genuinely entangled. It is interesting to find out an example, because it shows the connection between
UPBs and tripartite genuine entanglement.

\section*{Acknowledgments}
Authors were supported by the NNSF of China (Grant No. 11871089), and
the Fundamental Research Funds for the Central Universities (Grant No. ZG216S2005).

\appendix

\section{The proof of Theorem \ref{thm:11times7}}
\label{eq:pftheom6}
We begin by presenting the following Lemma \ref{le:submatrix} and \ref{le:pj}.
\begin{lemma}
	\label{le:submatrix}
	Suppose $A=[a_{i,j}], i=1,\cdots,11,j=1,\cdots,7$ is an $11\times7$ UOM. Then
	
	(i) the multiplicity of any element $a_{i,j}$ is at most four. 
	
	(ii) every column of $A$ has at least two and at most five independent elements.
	
	(iii) $A$ does not have the $6\times2$ submatrix $[a_{i,j}], i=1,\cdots,6,j=1,2$ such that $a_{1,1}=a_{2,1}=a_{3,1},a_{4,2}=a_{5,2}=a_{6,2}$ or $a_{1,1}=a_{2,1}=a_{3,1}=a_{4,1},a_{5,2}=a_{6,2}$.
	
	Up to permutation, we obtain that every entry in the lower-right $7\times6$ submatrix of $A$ has multiplicity one when $a_{1,1}=a_{2,1}=a_{3,1}=a_{4,1}$. So we have $\s(A_j)=4$ or $5$ for $j>1$. %Furthermore, when $a_{1,1}=a_{2,1}=a_{3,1}$, we obtain that the multiplicity of every entry in the lower right $8\times6$ submatrix of  $A$ has at most two. So we have $\s(A_j)=3,4$ or  $5$ for $j>1$. 
	
	(iv) if $A$ has the submatrix $\bma x&a_{1,2}\\\vdots&\vdots\\ x&a_{m,2}\\a_{m+1,1}&y\\\vdots&\vdots\\a_{m+n,1}&y\ema$, then there exists an integer $l\in[m+n+1,11]$ such that the row $\bma a_{l,1}&a_{l,2}\ema$ is orthogonal to the row $\bma x&a_{k,2}\ema$ or $\bma a_{k,1}&y\ema$, where $k\in[1,m+n]$.
	
	(v) if $\m(x)-\m(y)=m\geq0$, $x\in A_1,y\in 
	A_2$, then $A$ does not have a $(\m(x)+1)\times2$  submatrix 
	\begin{eqnarray}
		\label{eq:mxmym}
		\bma
		x&a_{1,2}\\
		\vdots&\vdots\\
		x&a_{m,2}\\
		x&y\\
		\vdots&\vdots\\
		x&y\\
		x'&y'	
		\ema.
	\end{eqnarray}
	
	(vi) $A$ does not have the $7\times3$ submatrix $[a_{i,j}],i=1,\cdots,7,j=1,2,3$ such that $a_{1,1}=a_{2,1}=a_{3,1}, a_{4,2}=a_{5,2}, a_{6,3}=a_{7,3}$.
	
	(vii) $A$ does not have the $8\times4$ submatrix $[a_{i,j}], i=1,\cdots,8,j=1,\cdots,4$ such that $a_{1,1}=a_{2,1},a_{3,2}=a_{4,2},a_{5,3}=a_{6,3},a_{7,4}=a_{8,4}$.
	%(viii) \tbc	
\end{lemma}

\begin{proof}
	(i) We prove the assertion by contradiction. Up to equivalence, suppose multiplicity of $a_{1,1}$ is five. Then we may assume that $a_{1,1}=a_{2,1}=a_{3,1}=a_{4,1}=a_{5,1}$. There exists a product vector $\ket{a_{1,1}', a_{6,2}',a_{7,3}',a_{8,4}',a_{9,5}',a_{10,6}',a_{11,7}'}$ such that it is orthogonal to all row vectors of $A$. It is a contradiction with the definition of UOM. 
	
	(ii) From (i), we have (ii) holds.
	
	(iii) When $a_{1,1}=a_{2,1}=a_{3,1}$ and $a_{4,2}=a_{5,2}=a_{6,2}$, there exists a product vector $\ket{a_{1,1}',a_{4,2}',a_{7,3}',a_{8,4}',a_{9,5}',a_{10,6}',a_{11,7}'}$ such that it is orthogonal to all row vectors of $A$. It is a contradition with the fact that $A$ is an $11\times7$ UOM. Similarly, we obtain $A$ does not have the $6\times2$ submatrix when $a_{1,1}=a_{2,1}=a_{3,1}=a_{4,1},a_{5,2}=a_{6,2}$. Because $A$ is an $11\times7$ UOM, when $a_{1,1}=a_{2,1}=a_{3,1}=a_{4,1}$, we obtain that $\s(A_j)=4$ or $5$ for $j>1$ up to permutation. %Similarly, when $a_{1,1}=a_{2,1}=a_{3,1}$, we obtain that the multiplicity of every entry in the lower-right $8\times6$ submatrix of $A$ has at most two. If there exists $\s(A_j)=2$, then we  . Thus, we obtain that $\s(A_j)=3,4$ or $5$ for $j>1$.  
	
	(iv) We prove the assertion by contradiction. There exists an integer $k\in[1,m+n]$ such that $\bma x&a_{k,2}\ema$ or $\bma a_{k,1}&y\ema$ is not orthogonal to $\bma a_{l_k,1}&a_{l_k,2}\ema$ for any $l_k\in[m+n+1,11]$. Then we obtain that $\bma a_{k,3}&a_{k,4}&a_{k,5}&a_{k,6}&a_{k,7}\ema$ is orthogonal to the lower-right $(11-m-n)\times5$ submatrix. We denote $u$ as the row $\bma a_{k,3}&a_{k,4}&a_{k,5}&a_{k,6}&a_{k,7}\ema$. Then we obtain that the row $\bma x'&y'&u\ema$ is orthogonal to $A$. It is a contradiction with the fact	that $A$ is a UOM.
	
	(v) In \eqref{eq:mxmym}, the equation $\m(x)-\m(y)=m\geq0$ implies  that the lower-left $(10-\m(x))\times2$ submatrix of $A$ has no $x$ and $y$. It implies that $\bma x'&y'\ema$ is not orthogonal to the lower-left $(10-\m(x))\times2$ submatrix. We denote $\bma x'&y'&v\ema$ as the row with $\bma x'&y'\ema$. Because $\bma x'&y'\ema$ is not orthogonal to the lower-left $(10-\m(x))\times2$ submatrix, we obtain that the vector $v$ is orthogonal to the $(10-\m(x))\times5$ submatrix. Hence, we obtain that $\bma x'&y&v\ema$ is orthogonal to $A$. It is a contradiction with the fact that $A$ is a UOM.
	
	(vi) When $a_{1,1}=a_{2,1}=a_{3,1}$ and $a_{4,2}=a_{5,2}, a_{6,3}=a_{7,3}$, there exists a product vector $\ket{a_{1,1}',a_{4,2}',a_{6,3}',a_{8,4}',a_{9,5}',a_{10,6}',a_{11,7}'}$ such that it is orthogonal to all row vectors of $A$. It is a contradition with the fact that $A$ is an $11\times7$ UOM.
	
	(vii) When $a_{1,1}=a_{2,1},a_{3,2}=a_{4,2},a_{5,3}=a_{6,3},a_{7,4}=a_{8,4}$, there exists a product vector $\ket{a_{1,1}',a_{3,2}',a_{5,3}',a_{7,4}',a_{9,5}',a_{10,6}',a_{11,7}'}$ such that it is orthogonal to all row vectors of $A$. It is a contradition with the fact that $A$ is an $11\times7$ UOM. 
	%	(viii) \tbc
\end{proof}

Using this lemma, we exclude some structures of $11\times7$ UOMs. In order to construct $11\times7$ UOMs, from the definition of $\s(A_j)$ above Eq. \eqref{eq:pj}, we investigate the properties about $\s(A_j)$ for each column $A_j$ as follows.

\begin{lemma}
	\label{le:pj}
	Suppose $A = [a_{i,j}], i=1,...,11, j=1,...,7$ is an $11\times7$ UOM. Then
	
	(i) 
	\begin{eqnarray}
		\s(A_j)&=&2\Longrightarrow p_j=12,13,14,15,16,18;\\
		\s(A_j)&=&3\Longrightarrow p_j=8,9,10,11,12,14;\\\label{eq:sigma4}
		\s(A_j)&=&4\Longrightarrow p_j=7,8,9;\\\label{eq:sigma5}
		\s(A_j)&=&5\Longrightarrow p_j=6,
	\end{eqnarray}
	where $p_j$ is from Eq. \eqref{eq:pj}.
	
	(ii) if there exists an entry $a_{i,j}$ such that $\m(a_{i,j})=4$, then  for other column $k\in\{1,\cdots,7\}\backslash\{j\}$ we obtain that $\s(A_k)=4$ or $5$.   
	
	(iii) there exist at least two columns $A_i,A_j$ such that $\s(A_i),\s(A_j)\leq4$. If there are five columns with $\s=5$, then the remaining two columns $A_i, A_j$ satisfy $\s(A_i)\leq\s(A_j)$ and $(\s(A_i),\s(A_j))=(2,2),(2,3),(2,4)$ or $(3,3)$. 
	
	(iv) $A$ has at most four columns each of which has four identical entries.
\end{lemma}

\begin{proof}
	(i) If there exists a column $j$ such that $\s(A_j)=2$, then we may assume that $\mu(a_{1,j})=4, \mu(a_{1,j}')=2$ and $\mu(a_{2,j})=4, \mu(a_{2,j}')=1$. By using \eqref{eq:pj},
	we obtain that
	\begin{eqnarray}
		\label{eq:sig2}
		p_j=\sum\m(x)\m(x')=\mu(a_{1,j})\mu(a_{1,j}')+\mu(a_{2,j})\mu(a_{2,j}')=4\times2+4\times1=12.
	\end{eqnarray}
	It is just one case for $\s(A_j)=2$. Moreover, we list other cases in the following.
	\begin{eqnarray}
		p_j=\sum\m(x)\m(x')=4\times1+3\times3=13,
	\end{eqnarray}
	or
	\begin{eqnarray}
		p_j=\sum\m(x)\m(x')=4\times2+3\times2=14,
	\end{eqnarray}
	or
	\begin{eqnarray}
		p_j=\sum\m(x)\m(x')=4\times3+3\times1\quad or\\ 3\times3+3\times2=15,
	\end{eqnarray}
	or
	\begin{eqnarray}
		p_j=\sum\m(x)\m(x')=4\times3+2\times2=16,
	\end{eqnarray}
	or
	\begin{eqnarray}
		p_j=\sum\m(x)\m(x')=4\times4+2\times1=18.
	\end{eqnarray}
	
	Similarly, when $\s(A_j)=3$, we obtain that
	\begin{eqnarray}
		p_j=4\times1+3\times1+1\times1 \quad or\\	4\times1+2\times1+2\times1\quad or\\
		3\times1+3\times1+2\times1=8,
	\end{eqnarray}
	or
	\begin{eqnarray}
		p_j=4\times1+2\times2+1\times1 \quad or\\	3\times1+2\times2+2\times1=9,
	\end{eqnarray}
	or
	\begin{eqnarray}
		p_j=3\times2+3\times1+1\times1 \quad or\\	3\times2+2\times1+2\times1\quad or\\
		2\times2+2\times2+2\times1=10,
	\end{eqnarray}
	or
	\begin{eqnarray}
		p_j=4\times2+2\times1+1\times1 \quad or\\	3\times2+2\times2+1\times1=11,
	\end{eqnarray}
	or
	\begin{eqnarray}
		p_j=3\times3+2\times1+1\times1=12,
	\end{eqnarray}
	or
	\begin{eqnarray}
		p_j=4\times3+1\times1+1\times1=14.
	\end{eqnarray}
	
	When $\s(A_j)=4$, we obtain that 
	\begin{eqnarray}
		\label{eq:pj7}
		p_j=4\times1+1\times1+1\times1+1\times1 \quad or\\	3\times1+2\times1+1\times1+1\times1\quad or\\
		2\times1+2\times1+2\times1+1\times1=7,
	\end{eqnarray}
	or
	\begin{eqnarray}
		p_j=2\times2+2\times1+1\times1+1\times1=8,
	\end{eqnarray}
	or
	\begin{eqnarray}
		\label{eq:pj9}
		p_j=3\times2+1\times1+1\times1+1\times1=9.
	\end{eqnarray}
	
	When $\s(A_j)=5$, we obtain that 
	\begin{eqnarray}
		\label{eq:sig5}
		p_j=2\times1+1\times1+1\times1+1\times1+1\times1=6.
	\end{eqnarray}
	
	(ii) Up to permutation, we may assume that the multiplicition of $a_{1,1}$ is four. According to Lemma \ref{le:submatrix} (i) and (iii), the multiplicity of every entry in lower-right $7\times6$ submatrix of $A$ is one.  Then we obtain that $\s(A_j)=4$ or $5$ in column $2-7$ of $A$. From \eqref{eq:sigma4} and \eqref{eq:sigma5}, we obtain that $p_k=6,7,8$ or $9$.
	
	(iii) From \eqref{eq:inequalityp}, we obtain that $m=11,n=7$ and $p_j\geq 55$. If there is no column such that $\s(A_j)\leq4, 1\leq j\leq7$, then we obtain that $\sum_jp_j=6+6+6+6+6+6+6=42<55$. If there exists only one column $A_i$ such that $\s(A_i)\leq4$, then we obtain that $\sum_jp_j\leq18+6+6+6+6+6+6=54<55$. So the two cases are contradictory with $p_j\geq55$. Then there exist at least two columns $A_i$ and $A_j$ such that $\s(A_i),\s(A_j)\leq4$. Furthermore, when five columns satisfy $\s=5$, $(\s(A_i),\s(A_j))=(2,2),(2,3),(2,4)$ or $(3,3)$. Otherwise, 
	up to permutation, we may assume that $\bma \s(A_1)&\s(A_2)&\s(A_3)&\s(A_4)&\s(A_5)&\s(A_6)&\s(A_7)\ema=\bma3&4&5&5&5&5&5\ema$ or $\bma4&4&5&5&5&5&5\ema$. Then $\sum_j p_j\leq14+9+6+6+6+6+6=53<55$ and $\sum_j p_j\leq9+9+6+6+6+6+6=48<55$, respectively. Thus, the two cases are contradictions with \eqref{eq:inequalityp}. 
	
	(iv) Up to permutation, suppose each of  $A_j,j=1,\cdots,5$ of $A$ has four identical entries. From Lemma \ref{le:submatrix} (iii), we obtain that $\s(A_i)=4$ or $5$ for any $i$. Then  $\sum_ip_i\leq5\times7+9\times2=53<55$. It is a contradiction with Lemma \eqref{eq:inequalityp}. Hence, we obtain that $A$ has at most four columns each of which has four identical entries.
\end{proof}

 By analysing Lemma \ref{le:submatrix} and \ref{le:pj}, we construct a concrete $11\times7$ UOM in Theorem \ref{thm:11times7}. Then we show its proof as follows.

\begin{proof}
	Denote $A$ as the $11\times7$ matrix in \eqref{eq:11times7uom} and $R_i$ as the $i$-th row of $A$, $i=1,\cdots,11$. From the definition of row orthogonal below Lemma \ref{le:functionp}, one can obtain that any two rows of $A$ are orthogonal. For each column $A_j, j=1,\cdots,7$ of $A$, we have $\bma p_1&p_2&p_3&p_4&p_5&p_6&p_7\ema=\bma 10&10&11&6&6&6&6\ema$. From \eqref{eq:11times7uom}, we obtain 
	\begin{eqnarray}
		\label{eq:p3}
		p_1&=&p_2=10=2\times2+2\times2+2\times1,\nonumber\\
		p_3&=&11=3\times2+2\times2+1\times1,\nonumber\\
		p_4&=&p_5=p_6=p_7=6=2\times1+1\times1+1\times1+1\times1+1\times1.
	\end{eqnarray}
	So $\sum_{j=1}^7 p_j=10+10+11+6+6+6+6=55$. From \eqref{eq:11times7uom} and \eqref{eq:p3}, we obtain that exactly entry $a_{1,3}$ of column $A_3$ has multiplicity three and entries of other columns have multiplicity at most two. Suppose $A$ is not a UOM. Then there exists a row $\bma b_1&b_2&b_3&b_4&b_5&b_6&b_7\ema$ orthogonal to $A$. %Denote $R_i$ as the $i$-th row of $A$, $i=1,\cdots,11$. 
	By considering $b_3$ in row  $\bma b_1&b_2&b_3&b_4&b_5&b_6&b_7\ema$, there exist two cases (i) $b_3=a_{1,3}'$ and (ii) $ b_3\neq a_{1,3}'$ as follows. 
	
	(i) If $b_3=a_{1,3}'$, then we obtain that the row $\bma b_1&b_2&a_{1,3}'&b_4&b_5&b_6&b_7\ema$ is orthogonal to  $R_1,R_3$, $R_6$ and $\{b_i\}_{i\in\{1,\cdots,7\}\backslash\{3\}}$ must be  orthogonal to remaining eight rows of $A$. So there exist any two entries of $\{b_i\}_{i\in\{1,\cdots,7\}\backslash\{3\}}$  orthogonal to at least four rows of $A$. Because any entries of $A_j, j\in \{1,\cdots,7\}\backslash\{3\}$ have at most multiplicity two, we obtain that there exists a $7\times3$ submatrix of $A$ equivalent to $\mathcal{F}_1$ up to permutation, where $\mathcal{F}_1:=
	\bma
	x_1&*&*\\
	x_1&*&*\\
	x_1&*&*\\
	*&x_2&*\\
	*&x_2&*\\
	*&*&x_3\\
	*&*&x_3
	\ema$. 
	%We assume that $A$ has the $7\times3$ submatrix at first,  
	We denote $M\sim N$ as equivalent matrix $M$ and $N$.  Namely, there exist invertible product matrix $U$ and $V$ such that $UMV = N$. 
	Up to row permutation of $A$, we have 
	\begin{eqnarray}
		\label{eq:}
		A\sim A'=\bma
		a_{1,1}&a_{1,2}&a_{1,3}&a_{1,4}&a_{1,5}&a_{1,6}&a_{1,7}\\
		a_{1,1}'&a_{2,2}&a_{1,3}&a_{3,4}&a_{3,5}&a_{3,6}&a_{3,7}\\
		a_{4,1}&a_{1,2}'&a_{1,3}&a_{3,4}'&a_{6,5}&a_{2,6}'&a_{5,7}'\\
		a_{1,1}&a_{2,2}&a_{2,3}&a_{2,4}&a_{1,5}'&a_{2,6}&a_{2,7}\\
		a_{1,1}'&a_{1,2}&a_{2,3}&a_{4,4}&a_{4,5}&a_{4,6}&a_{3,7}'\\
		a_{4,1}&a_{1,2}'&a_{5,3}&a_{2,4}'&a_{5,5}&a_{3,6}'&a_{5,7}\\
		a_{4,1}'&a_{7,2}&a_{2,3}'&a_{7,4}&a_{3,5}'&a_{1,6}'&a_{7,7}\\
		a_{4,1}'&a_{7,2}'&a_{2,3}'&a_{3,4}'&a_{6,5}&a_{8,6}&a_{1,7}'\\
		a_{9,1}&a_{7,2}'&a_{1,3}'&a_{4,4}'&a_{5,5}'&a_{8,6}'&a_{2,7}'\\
		a_{9,1}&a_{7,2}&a_{1,3}'&a_{7,4}'&a_{4,5}'&a_{2,6}'&a_{5,7}'\\
		a_{9,1}'&a_{2,2}'&a_{5,3}'&a_{1,4}'&a_{6,5}'&a_{4,6}'&a_{7,7}'
		\ema.
	\end{eqnarray}	
	One can verify that entries in the lower-right $8\times4$ submatrix of $A'$ have multiplicity one. So the entries of $A'$ with multiplicity two must exist in the lower-left $8\times2$ submatrix of $A'$. Because   there exists a $7\times3$ submatrix of $A$ equivalent $\mathcal{F}_1$,  two entries with multiplicity two in the lower-left $8\times2$ submatrix of $A'$ are not in the same row. In the lower-left $8\times2$ submatrix of $A'$, the entries with multiplicity two are $a_{4,1}',a_{9,1}, a_{7,2}$ and $a_{7,2}'$, respectively. They are in the $4\times2$ submatrix $\bma a_{4,1}'&a_{7,2}\\
	a_{4,1}'&a_{7,2}'\\
	a_{9,1}&a_{7,2}'\\
	a_{9,1}&a_{7,2}\ema$ of $A'$. However, the entries $a_{4,1}',a_{9,1}$ are in the same row with the entries $a_{7,2}$ and $a_{7,2}'$. So it is a contradiction with the fact that two entries with multiplicity two are not in the same row. Thus the $7\times3$ submatrix does not exist in $A$. We have excluded case (i).

	(ii) If $b_3\neq a_{1,3}'$, then there must exist four entries of $\bma b_1&b_2&b_3&b_4&b_5&b_6&b_7\ema$ such that a row containing the four entries is orthogonal to at least eight rows of $A$. Up to permutation, there exists an $8\times4$ submatrix $X_{A'}$ of $A$ equivalent to $\mathcal{F}_2$, where $\mathcal{F}_2:=\bma
	y_1&*&*&*\\
	y_1&*&*&*\\
	*&y_2&*&*\\
	*&y_2&*&*\\
	*&*&y_3&*\\
	*&*&y_3&*\\
	*&*&*&y_4\\
	*&*&*&y_4
	\ema$. %Thus, we obtain that $A$ has the $7\times3$ or $8\times4$ submatrix that it is equivalent to $\mathcal{F}_1$ or $\mathcal{F}_2$, respectively. 
	% if there exists a product vector such that it is orthogonal to $A$, then $A$ has the $7\times3$ or $8\times4$ submatrix such that it violates Lemma \ref{le:submatrix} (vi) and (vii). Because in the third column, the multipli
	%we assume that $A$ has the $8\times4$ submatrix such that it is equivalent to $\mathcal{F}_2$. 
	In column $A_4,A_5,A_6,A_7$ of $A$, we obtain that entries with multiplicity two are $a_{3,4}',a_{6,5}, a_{2,6}'$ and $a_{5,7}'$, respectively. From \eqref{eq:11times7uom}, we obtain that the entries $a_{3,4}',a_{6,5}$ and $a_{2,6}', a_{5,7}$ are in row ${R}_6, {R}_{8}$ and ${R}_6, {R}_{10}$ of $A$, respectively.  Because the entries $a_{3,4}',a_{6,5}, a_{2,6}', a_{5,7}'$ are in ${R}_6$,  we obtain that $X_{A'}$ is a submatrix of $\bma A_1&A_2&A_3&A_k\ema$, where $k\in\{4,5,6,7\}$. Because $a_{3,4}',a_{6,5}$ are in $R_6,R_8$ and $a_{2,6}',a_{5,7}$ are in $R_6,R_{10}$,  we obtain that there exist two cases (ii.a) $k=4$ or $5$ and (ii.b) $k=6$ or $7$. 
	
	(ii.a) When $k=4$ or $5$, up to row permutation of $A$, we have 
	\begin{eqnarray}
		A\sim B_1=	\bma
		a_{4,1}&a_{1,2}'&a_{1,3}&a_{3,4}'&a_{6,5}&a_{2,6}'&a_{5,7}'\\
		a_{4,1}'&a_{7,2}'&a_{2,3}'&a_{3,4}'&a_{6,5}&a_{8,6}&a_{1,7}'\\
		a_{1,1}&a_{1,2}&a_{1,3}&a_{1,4}&a_{1,5}&a_{1,6}&a_{1,7}\\
		a_{1,1}&a_{2,2}&a_{2,3}&a_{2,4}&a_{1,5}'&a_{2,6}&a_{2,7}\\
		a_{1,1}'&a_{2,2}&a_{1,3}&a_{3,4}&a_{3,5}&a_{3,6}&a_{3,7}\\
		a_{1,1}'&a_{1,2}&a_{2,3}&a_{4,4}&a_{4,5}&a_{4,6}&a_{3,7}'\\
		a_{4,1}&a_{1,2}'&a_{5,3}&a_{2,4}'&a_{5,5}&a_{3,6}'&a_{5,7}\\
		a_{4,1}'&a_{7,2}&a_{2,3}'&a_{7,4}&a_{3,5}'&a_{1,6}'&a_{7,7}\\
		a_{9,1}&a_{7,2}'&a_{1,3}'&a_{4,4}'&a_{5,5}'&a_{8,6}'&a_{2,7}'\\
		a_{9,1}&a_{7,2}&a_{1,3}'&a_{7,4}'&a_{4,5}'&a_{2,6}'&a_{5,7}'\\
		a_{9,1}'&a_{2,2}'&a_{5,3}'&a_{1,4}'&a_{6,5}'&a_{4,6}'&a_{7,7}'
		\ema.
	\end{eqnarray}
	Then in $B_1$, a row containing entries $b_1,b_2,b_3$ of $\bma b_1&b_2&b_3&b_4&b_5&b_6&b_7\ema$ is orthogonal to at least six rows of the lower-left $9\times3$ submatrix $X_{B_1}$ of $B_1$. Then  there should exist three entries with multiplicity two in six different rows of $X_{B_1}$. In $X_{B_1}$, the entries with multiplicity two are $a_{1,1},a_{1,1}',a_{9,1}, a_{1,2}, a_{2,2}, a_{7,2}, a_{1,3}, a_{2,3}, a_{1,3}'$. They are in a $7\times3$ submatrix $B_1':=\bma a_{1,1}&a_{1,2}&a_{1,3}\\
	a_{1,1}&a_{2,2}&a_{2,3}\\
	a_{1,1}'&a_{2,2}&a_{1,3}\\
	a_{1,1}'&a_{1,2}&a_{2,3}\\
	*&a_{7,2}&*\\
	a_{9,1}&*&a_{1,3}'\\
	a_{9,1}&a_{7,2}&a_{1,3}'\\
	\ema$ of $B_1$. 
	So we obtain a fact that there should exist a $6\times3$ submatrix of $B_1'$ equivalent to a submatrix 
	$\bma 
	y_1&*&*\\
	y_1&*&*\\
	*&y_2&*\\
	*&y_2&*\\
	*&*&y_3\\
	*&*&y_3
	\ema$ of $\mathcal{F}_2$ up to permutation. 
	However, in the first four rows of $B_1'$, there are no two entries with multiplicity two in different rows. It is a contradiction with the fact that we mentioned above. %there exists a $6\times3$ submatrix of $B_1'$ like a submatrix $\bma 
	%y_1&*&*\\
	%y_1&*&*\\
	%*&y_2&*\\
	%*&y_2&*\\
	%*&*&y_3\\
	%*&*&y_3
	%\ema$ of $\mathcal{F}_2$.  
	
	(ii.b) When $k=6$ or $7$, up to row permutation of $A$, we have 
	\begin{eqnarray}
		A\sim	B_2=\bma
		a_{4,1}&a_{1,2}'&a_{1,3}&a_{3,4}'&a_{6,5}&a_{2,6}'&a_{5,7}'\\
		a_{9,1}&a_{7,2}&a_{1,3}'&a_{7,4}'&a_{4,5}'&a_{2,6}'&a_{5,7}'\\
		a_{1,1}&a_{1,2}&a_{1,3}&a_{1,4}&a_{1,5}&a_{1,6}&a_{1,7}\\
		a_{1,1}&a_{2,2}&a_{2,3}&a_{2,4}&a_{1,5}'&a_{2,6}&a_{2,7}\\	a_{1,1}'&a_{2,2}&a_{1,3}&a_{3,4}&a_{3,5}&a_{3,6}&a_{3,7}\\
		a_{1,1}'&a_{1,2}&a_{2,3}&a_{4,4}&a_{4,5}&a_{4,6}&a_{3,7}'\\
		a_{4,1}&a_{1,2}'&a_{5,3}&a_{2,4}'&a_{5,5}&a_{3,6}'&a_{5,7}\\
		a_{4,1}'&a_{7,2}&a_{2,3}'&a_{7,4}&a_{3,5}'&a_{1,6}'&a_{7,7}\\
		a_{4,1}'&a_{7,2}'&a_{2,3}'&a_{3,4}'&a_{6,5}&a_{8,6}&a_{1,7}'\\
		a_{9,1}&a_{7,2}'&a_{1,3}'&a_{4,4}'&a_{5,5}'&a_{8,6}'&a_{2,7}'\\
		a_{9,1}'&a_{2,2}'&a_{5,3}'&a_{1,4}'&a_{6,5}'&a_{4,6}'&a_{7,7}'
		\ema.
	\end{eqnarray}
	Similarly, there should exist three entries with multiplicity two in the different rows of the lower-left $9\times3$ submatrix of $B_2$. By using the same method in case (ii.a), we obtain that there exist no such entries. Then the $8\times4$ submatrix does not exist in $A$. So case (ii) has been excluded. 
	
	Hence % the $8\times4$ submatrix does not exist in $A$. %there exists exactly one column of $A_4,\cdots,A_7$ in the $8\times4$ submatrix.
	$A$ does not have the $7\times3$ or $8\times4$ submatrix such that it is equivalent to $\mathcal{F}_1$ or $\mathcal{F}_2$, respectively. It is a contradiction with the above assumption below \eqref{eq:p3} that $A$ is not a UOM. Thus, $A$ is an $11\times7$ UOM. 
\end{proof}
\section{The proof of Theorem \ref{thm:distinguish}}
We first introduce some definitions of distinguishability. A set of orthogonal states is locally distinguishable if there exists a sequence of LOCC distinguishing the states. A measurement performed to distinguish a set of orthogonal states is called an orthogonality-preserving measurement if the states remain orthogonal after the measurement. Local distinguishability sufficiently ensures local reducibility \cite{PhysRevLett.122.040403}. A set of orthogonal quantum states is called a locally reducible set if it is possible to distinguish one or more states from the set by orthogonality-preserving local measurement. Then we show the proof of Theorem \ref{thm:distinguish}.

\label{pf:theorem7}
\begin{proof}	
	From Theorem  \ref{thm:11times7}, we have eleven orthogonal product vectors. We list them as follows.
	
	%Then we 
	\begin{eqnarray}
		\label{eq:product}
		\ket{\varphi_1}=\ket{a_{1,1},a_{1,2},a_{1,3},a_{1,4},a_{1,5},a_{1,6},a_{1,7}},\nonumber
		&&
		\\
		\ket{\varphi_2}=\ket{a_{1,1},a_{2,2},a_{2,3},a_{2,4},a_{1,5}',a_{2,6},a_{2,7}},\nonumber
		&&
		\\
		\ket{\varphi_3}=\ket{a_{1,1}',a_{2,2},a_{1,3},a_{3,4},a_{3,5},a_{3,6},a_{3,7}},\nonumber
		&&
		\\
		\ket{\varphi_4}=\ket{	a_{1,1}',a_{1,2},a_{2,3},a_{4,4},a_{4,5},a_{4,6},a_{3,7}'},\nonumber
		&&
		\\
		\ket{\varphi_5}=\ket{a_{4,1},a_{1,2}',a_{5,3},a_{2,4}',a_{5,5},a_{3,6}',a_{5,7}},\nonumber
		&&
		\\
		\ket{\varphi_6}=\ket{a_{4,1},a_{1,2}',a_{1,3},a_{3,4}',a_{6,5},a_{2,6}',a_{5,7}'},\nonumber
		&&
		\\
		\ket{\varphi_7}=\ket{a_{4,1}',a_{7,2},a_{2,3}',a_{7,4},a_{3,5}',a_{1,6}',a_{7,7}},\nonumber
		&&
		\\
		\ket{\varphi_8}=\ket{	a_{4,1}',a_{7,2}',a_{2,3}',a_{3,4}',a_{6,5},a_{8,6},a_{1,7}'},\nonumber
		&&
		\\
		\ket{\varphi_9}=\ket{	a_{9,1},a_{7,2}',a_{1,3}',a_{4,4}',a_{5,5}',a_{8,6}',a_{2,7}'},\nonumber
		&&
		\\
		\ket{\varphi_{10}}=\ket{	a_{9,1},a_{7,2},a_{1,3}',a_{7,4}',a_{4,5}',a_{2,6}',a_{5,7}'},\nonumber
		&&
		\\
		\ket{\varphi_{11}}=\ket{	a_{9,1}',a_{2,2}',a_{5,3}',a_{1,4}',a_{6,5}',a_{4,6}',a_{7,7}'}.
	\end{eqnarray} 
	Denote the set of these orthogonal product vectors as $\Phi$. We may assume that any $\ket{\varphi_i}\in\mathcal{H}_{C_1}\otimes\cdots\otimes\mathcal{H}_{C_7}=(\mathbb{C}^{2})^{\otimes7}$. Then we consider whether the set $\Phi$ is locally reducible when every $\ket{\varphi_i}$ is a bipartite product vector. Let $\mathcal{S}\subset\{1,\cdots,7\}$ be a subset. Denote $\overline{\mathcal{S}}$ as the complement of $\mathcal{S}$. Define the composite system as $C_{\mathcal{S}}:=\otimes_{j\in\mathcal{S}} C_j$  supported on the space $\otimes_{j\in\mathcal{S}}\mathcal{H}_{C_j}$ and $C_{\overline{\mathcal{S}}}:=\otimes_{j\in\overline{\mathcal{S}}} C_j$. If the set $\Phi$ is locally reducible, then the set $\Phi$ can be splited into two disjoint subsets $P$ and $Q$ such that the elements of $P$ are orthogonal to that of $Q$ on system $C_{\mathcal{S}}$, where $P\cup Q=\Phi$. %Without loss of generality, we may assume $|P|<|Q|$. 
	Denote $I_{C_j}$ as an identity matrix on system $C_j$. We consider every $\ket{\varphi_i}$ as a bipartite state on systems $C_\mathcal{S}$ and  $C_{\overline{\mathcal{S}}}$. 
	
	When $\mathcal{S}=\{m,n\}$, we have $\overline{\mathcal{S}}=\{1,\cdots,7\}\backslash\{m,n\}$, where $m,n\in\{1,\cdots,7\}$ and $m\neq n$. We can write every $\ket{\varphi_i}$ as
	\begin{eqnarray}
		\label{eq:varphi}
		\ket{\varphi_i}:=\ket{a_i}_{C_{\overline{\mathcal{S}}}}\otimes\ket{b_i,c_i}_{C_\mathcal{S}}, \quad i\in\{1,\cdots,11\}.
	\end{eqnarray} 
	
	We assume that the set $\Phi$ is locally distinguishable on systems $C_\mathcal{S}$ and  $C_{\overline{\mathcal{S}}}$. Then it is locally reducible. We obtain that the set $\Phi$ can be splitted into two disjoint subsets $P$ and $Q$ on system $C_{\mathcal{S}}$ or $C_{\overline{\mathcal{S}}}$. 
	Denote $|P|$ and $|Q|$ as the number of the elements in sets $P$ and $Q$, respectively. From \eqref{eq:product}, we obtain that $P\cup Q=\Phi$ and $|P|+|Q|=11$. Without loss of generality, we may assume that $|P|<|Q|$. In the following, we show two cases (i) and (ii).  
	
	(i) We assume that the set $\Phi$ is splitted into two disjoint subsets $P$ and $Q$ on system $C_{\mathcal{S}}$. It implies that the elements of the set $P$ are orthogonal to the elements of the set $Q$ on system $C_{\mathcal{S}}$. From \eqref{eq:p3} and \eqref{eq:varphi}, we obtain that $\ket{\varphi_i}$ is orthogonal to at most five other product vectors in \eqref{eq:product} on system $C_{\mathcal{S}}$. It implies $\max\{|P|,|Q|\}\leq5$. We obtain that $|P|+|Q|=10$. It is a contradiction with the fact that 
	$|P|+|Q|=11$. Thus set $\Phi$ is not splitted into two disjoint subsets $P$ and $Q$ on system $C_{\mathcal{S}}$.
	
	(ii) We assume that the set $\Phi$ is splitted into two disjoint subsets $P$ and $Q$ on system $C_{\overline{\mathcal{S}}}$. Since $|P|+|Q|=11$ and $|P|\leq|Q|$, we obtain $|P|\in\{1,2,3,4,5\}$. We consider the number of the elements of the set $P$ in the five following cases (ii.a) - (ii.e). 
	
	(ii.a) When $|P|=1$, we obtain that there exists $\ket{\varphi_i}$ in the set $P$ and remaining product vector $\ket{\varphi_j}$ in the set $Q$, where $j\in\{1,\cdots,11\}\backslash\{i\}$. From \eqref{eq:varphi}, there exists  $\ket{a_i}_{C_{\overline{\mathcal{S}}}}$ is orthogonal to the remaining $\ket{a_j}_{C_{\overline{\mathcal{S}}}}$, where $j\in\{1,\cdots,11\}\backslash\{i\}$. From the definition of UOM, we obtain that orthogonal product vectors $\ket{b_i}$ and $\ket{b_i'}$ exist in $\ket{\varphi_i}$ and $\ket{\varphi_j}$, respectively. There exists $\ket{\varphi_j}$ such that it is orthogonal to $\ket{\varphi_i}$ on system $C_{\mathcal{S}}$. Since there exists exactly one orthogonal pair in product vectors in $\eqref{eq:varphi}$, we obtain that $\ket{a_i}_{C_{\overline{\mathcal{S}}}}$ is not orthogonal to $\ket{a_j}_{C_{\overline{\mathcal{S}}}}$. It is a contradiction with the fact that there exists  $\ket{a_i}_{C_{\overline{\mathcal{S}}}}$ is orthogonal to the remaining $\ket{a_j}_{C_{\overline{\mathcal{S}}}}$.
	
	(ii.b) When $|P|=2$, we obtain that there exist two $\ket{\varphi_i}$ and $\ket{\varphi_j}$ in the set $P$ and remaining product vectors are in the set $Q$. From \eqref{eq:varphi}, there exist  $\ket{a_i}_{C_{\overline{\mathcal{S}}}}$ and $\ket{a_j}_{C_{\overline{\mathcal{S}}}}$ are orthogonal to the remaining $\ket{a_k}_{C_{\overline{\mathcal{S}}}}$, where $i,j\in\{1,\cdots,11\}$ and $k\in\{1,\cdots,11\}\backslash\{i,j\}$. We obtain that $\ket{\varphi_i}$ and $\ket{\varphi_j}$ are not orthogonal on system $C_{\overline{\mathcal{S}}}$. It implies that $\ket{a_i}$ is not orthogonal to $\ket{a_j}$. Otherwise, when $\ket{a_i}$ is orthogonal to $\ket{a_j}$, it is equivalent to the case of (ii.a). From  \eqref{eq:p3}, we obtain that any two $\ket{a_i}$ and $\ket{a_j}$ are linearly independent when they are not orthogonal. Because $\ket{\varphi_i}$ and $\ket{\varphi_j}$ are orthogonal on system $C_{\mathcal{S}}$, From \eqref{eq:varphi}, we obtain that $\ket{b_i,c_i}$ and $\ket{b_j,c_j}$ are orthogonal. Since any two product vectors has exactly one orthogonal pair, it implies that $\ket{b_i}, \ket{b_j}$ or $\ket{c_i},\ket{c_j}$ are orthogonal. If we assume that $\ket{b_i}, \ket{b_j}$ ($\ket{c_i},\ket{c_j}$) are orthogonal, then we obtain that $\ket{c_i},\ket{c_j}$ ($\ket{b_i}, \ket{b_j}$) are not orthogonal. We obtain that there exists a product vector $\ket{\varphi_k}$ such that $\ket{c_k} (\ket{b_k})$ is orthogonal to $\ket{c_i} (\ket{b_i})$. It implies that $\ket{a_i}$ is not orthogonal to $\ket{a_k}$. So it violates the fact that $\ket{\varphi_i}$ and $\ket{\varphi_j}$ are not orthogonal on system $C_{\overline{\mathcal{S}}}$. 
	
	(ii.c) When $|P|=3$, we obtain that there exist three $\ket{\varphi_i}$, $\ket{\varphi_j}$ and $\ket{\varphi_k}$ in the set $P$ and remaining product vectors are in the set $Q$. Because $\ket{\varphi_i}$, $\ket{\varphi_j}$ and $\ket{\varphi_k}$ of the set $P$ are orthogonal to any $\ket{\varphi_l}$ of the set $Q$ on system $C_{\overline{{\mathcal{S}}}}$, they are not orthogonal on system $C_{\mathcal{S}}$. From \eqref{eq:varphi}, it implies that $\ket{b_i,c_i},\ket{b_j,c_j}$ and $\ket{b_k,c_k}$ are not orthogonal to any $\ket{b_l,c_l}$, where $l\in\{1,\cdots,11\}\backslash\{i,j,k\}$. To satisfy this condition, we obtain that one of $\ket{b_i}, \ket{b_j}$ and $\ket{b_k}$ is orthogonal to the other. Without loss of generality, we assume $\ket{b_i}$ is orthogonal to $\ket{b_j}$ and $\ket{b_k}$. Then we have $\ket{b_j}=\ket{b_k}=\ket{b_i'}$. Similarly, we may assume that $\ket{c_i}$ is orthogonal to $\ket{c_j}$ and $\ket{c_k}$. Then we have $\ket{c_j}=\ket{c_k}=\ket{c_i'}$. So there are four orthogonal pairs in $\ket{\varphi_i}$, $\ket{\varphi_j}$ and $\ket{\varphi_k}$ of the set $P$. It violates the fact that any two product vectors have exactly one orthogonal pair.
	
	(ii.d) When $|P|=4$, we obtain that there exist four $\ket{\varphi_{m_1}},\ket{\varphi_{m_2}},\ket{\varphi_{m_3}}$ and $\ket{\varphi_{m_4}}$ in the set $P$ and remaining product vectors are in the set $Q$. Because $\ket{\varphi_{m_1}},\ket{\varphi_{m_2}},\ket{\varphi_{m_3}}$ and $\ket{\varphi_{m_4}}$ of the set $P$ are orthogonal to any $\ket{\varphi_l}$ of the set $Q$ on system $C_{\overline{{\mathcal{S}}}}$, they are not orthogonal on system $C_{\mathcal{S}}$. From \eqref{eq:varphi}, it implies that $\ket{b_{m_1},c_{m_1}},\ket{b_{m_2},c_{m_2}}, \ket{b_{m_3},c_{m_3}}$ and $\ket{b_{m_4},c_{m_4}}$ are not orthogonal to any $\ket{b_l,c_l}$, where $l\in\{1,\cdots,11\}\backslash\{m_1,m_2,m_3,m_4\}$. To satisfy this condition, we have four cases (a), (b), (c), (d) about $\ket{b_{m_1},c_{m_1}},\ket{b_{m_2},c_{m_2}}, \ket{b_{m_3},c_{m_3}}, \ket{b_{m_4},c_{m_4}}$. For convenience of description, we denote $\{\ket{b_{m_1},c_{m_1}},\ket{b_{m_2},c_{m_2}}, \ket{b_{m_3},c_{m_3}}, \ket{b_{m_4},c_{m_4}}\}$ as a $4\times2$ matrix 
	\begin{eqnarray}
		\label{eq:matrix}
		\bma b_{m_1}&c_{m_1}\\
		b_{m_2}&c_{m_2}\\
		b_{m_3}&c_{m_3}\\
		b_{m_4}&c_{m_4}
		\ema.
	\end{eqnarray}

	(a) Suppose there exist two elements of $\ket{b_{m_1}},\ket{b_{m_2}},\ket{b_{m_3}},\ket{b_{m_4}}$ and two elements of $\ket{c_{m_1}},\ket{c_{m_2}},\ket{c_{m_3}},\ket{c_{m_4}}$ are same, respectively. Then we obtain that they are orthogonal to the remaining two elements, respectively. It implies that the  remaining two elements are same.  So we obtain that the matrix in \eqref{eq:matrix} has eight orthogonal pairs. From \eqref{eq:p3}, we obtain that the matrix in \eqref{eq:matrix} is a submatrix of the left $11\times3$ submatrix in \eqref{eq:11times7uom}. Because two elements of $\ket{b_{m_1}},\ket{b_{m_2}},\ket{b_{m_3}},\ket{b_{m_4}}$ are same and orthogonal to remaining two elements, there exist  five $4\times3$ submatrices 
	\begin{eqnarray}
		\bma
		a_{1,1}&a_{1,2}&a_{1,3}\\
		a_{1,1}&a_{2,2}&a_{2,3}\\
		a_{1,1}'&a_{1,2}&a_{1,3}\\
		a_{1,1}'&a_{2,2}&a_{2,3}
		\ema,\quad
		\bma
		a_{4,1}&a_{1,2}'&a_{5,3}\\
		a_{4,1}&a_{1,2}'&a_{1,3}\\
		a_{4,1}'&a_{7,2}&a_{2,3}'\\
		a_{4,1}'&a_{7,2}'&a_{2,3}'
		\ema,\quad
		\bma
		a_{1,1}&a_{1,2}&a_{1,3}\\
		a_{1,1}'&a_{1,2}&a_{1,3}\\
		a_{4,1}&a_{1,2}'&a_{5,3}\\
		a_{4,1}&a_{1,2}'&a_{1,3}
		\ema,\nonumber
	\end{eqnarray}
	\begin{eqnarray}
		\bma
		a_{4,1}'&a_{7,2}&a_{2,3}'\\
		a_{9,1}&a_{7,2}&a_{1,3}'\\
		a_{4,1}'&a_{7,2}'&a_{2,3}'\\
		a_{9,1}&a_{7,2}'&a_{1,3}'
		\ema,\quad
		\bma
		a_{1,1}&a_{2,2}&a_{2,3}\\
		a_{1,1}'&a_{1,2}&a_{2,3}\\
		a_{4,1}'&a_{7,2}&a_{2,3}'\\
		a_{4,1}'&a_{7,2}'&a_{2,3}'
		\ema.
	\end{eqnarray}
	Because every $4\times2$ submatrix in the above matrices has at most six orthogonal pairs, it violate the fact that the matrix in \eqref{eq:matrix} has eight orthogonal pairs. So this case does not hold.
	
	(b) Suppose there exist two elements of $\ket{b_{m_1}},\ket{b_{m_2}},\ket{b_{m_3}},\ket{b_{m_4}}$ are same and $\ket{c_{m_1}},\ket{c_{m_2}},\ket{c_{m_3}},\ket{c_{m_4}}$ have two orthogonal pairs. We obtain that the matrix in \eqref{eq:matrix} has six orthogonal pairs. From \eqref{eq:11times7uom} and \eqref{eq:p3}, we obtain that the $4\times2$ matrix must be the submatrix of the five $4\times5$ matrices
	\begin{eqnarray}
		\bma
		a_{1,1}&a_{1,4}&a_{1,5}&a_{1,6}&a_{1,7}\\
		a_{1,1}&a_{2,4}&a_{1,5}'&a_{2,6}&a_{2,7}\\
		a_{1,1}'&a_{3,4}&a_{3,5}&a_{3,6}&a_{3,7}\\	a_{1,1}'&a_{4,4}&a_{4,5}&a_{4,6}&a_{3,7}'\\
		\ema,\quad
		\bma
		a_{4,1}&a_{2,4}'&a_{5,5}&a_{3,6}'&a_{5,7}\\
		a_{4,1}&a_{3,4}'&a_{6,5}'&a_{2,6}'&a_{5,7}'\\
		a_{4,1}'&a_{7,4}&a_{3,5}'&a_{1,6}'&a_{7,7}\\	a_{4,1}'&a_{3,4}'&a_{6,5}&a_{8,6}&a_{1,7}'\\
		\ema,\quad
		\bma
		a_{1,2}&a_{1,4}&a_{1,5}&a_{1,6}&a_{1,7}\\
		a_{1,2}&a_{4,4}&a_{4,5}&a_{4,6}&a_{3,7}'\\
		a_{1,2}'&a_{2,4}'&a_{5,5}&a_{3,6}'&a_{5,7}\\	a_{1,2}'&a_{3,4}'&a_{6,5}&a_{2,6}'&a_{5,7}'\\
		\ema,
		\notag
	\end{eqnarray}
	\begin{eqnarray}
		\bma
		a_{7,2}&a_{7,4}&a_{3,5}'&a_{1,6}'&a_{7,7}\\
		a_{7,2}&a_{7,4}'&a_{4,5}'&a_{2,6}'&a_{5,7}'\\
		a_{7,2}'&a_{3,4}'&a_{6,5}&a_{8,6}&a_{1,7}'\\	a_{7,2}'&a_{4,4}'&a_{5,5}'&a_{8,6}'&a_{2,7}'\\
		\ema,\quad
		\bma
		a_{2,3}&a_{2,4}'&a_{5,5}&a_{3,6}'&a_{5,7}\\
		a_{2,3}&a_{3,4}'&a_{6,5}'&a_{2,6}'&a_{5,7}'\\
		a_{2,3}'&a_{7,4}&a_{3,5}'&a_{1,6}'&a_{7,7}\\	a_{2,3}'&a_{3,4}'&a_{6,5}&a_{8,6}&a_{1,7}'\\
		\ema.
	\end{eqnarray}
	Then we obtain that every $4\times2$ submatrix of the above $4\times5$ matrices has at most five orthogonal pairs. It is a contradiction with the fact that there exists a $4\times2$ matrix such that  it has six orthogonal pairs. So this case does not hold.
	
	(c) Suppose there exist two elements of $\ket{c_{m_1}},\ket{c_{m_2}},\ket{c_{m_3}},\ket{c_{m_4}}$ are same and $\ket{b_{m_1}},\ket{b_{m_2}},\ket{b_{m_3}},\ket{b_{m_4}}$ have two orthogonal pairs. This case is similar to (b). So this case does not hold.
	
	(d) Suppose $\ket{b_{m_1}},\ket{b_{m_2}},\ket{b_{m_3}},\ket{b_{m_4}}$ and $\ket{c_{m_1}},\ket{c_{m_2}},\ket{c_{m_3}},\ket{c_{m_4}}$ both have two pairs orthogonal pairs.  Then we obtain that the $4\times2$ matrix in  \eqref{eq:matrix} has four orthogonal pairs. 	
	From  \eqref{eq:11times7uom} and \eqref{eq:p3}, by permutating rows and columns, we obtain that the $4\times2$ matrix must be in the submatrix of the six $11\times2$ matrices
	\begin{eqnarray}
		\bma
		a_{1,4}&a_{1,5}\\
		a_{2,4}&a_{1,5}'\\	
		a_{3,4}&a_{3,5}\\
	a_{4,4}&a_{4,5}\\
		a_{2,4}'&a_{5,5}\\
		a_{3,4}'&a_{6,5}\\
		a_{7,4}&a_{3,5}'\\
		a_{3,4}'&a_{6,5}\\
	a_{4,4}'&a_{5,5}'\\
		a_{7,4}'&a_{4,5}'\\
		a_{1,4}'&a_{6,5}'						
		\ema,
		\bma
	a_{1,4}&a_{1,6}\\
		a_{2,4}&a_{2,6}\\	
		a_{3,4}&a_{3,6}\\
	a_{4,4}&a_{4,6}\\
	a_{2,4}'&a_{3,6}'\\
		a_{3,4}'&a_{2,6}'\\
		a_{7,4}&a_{1,6}'\\
	a_{3,4}'&a_{8,6}\\
	a_{4,4}'&a_{8,6}'\\
	a_{7,4}'&a_{2,6}'\\
	a_{1,4}'&a_{4,6}'
		\ema,
		\bma
	a_{1,4}&a_{1,7}\\
	a_{2,4}&a_{2,7}\\	
	a_{3,4}&a_{3,7}\\
	a_{4,4}&a_{3,7}'\\
	a_{2,4}'&a_{5,7}\\
	a_{3,4}'&a_{5,7}'\\
	a_{7,4}&a_{7,7}\\
	a_{3,4}'&a_{1,7}'\\
	a_{4,4}'&a_{2,7}'\\
	a_{7,4}'&a_{5,7}'\\
	a_{1,4}'&a_{7,7}'	
		\ema,
		\bma
		a_{1,5}&a_{1,6}\\
		a_{1,5}'&a_{2,6}\\	
		a_{3,5}&a_{3,6}\\
		a_{4,5}&a_{4,6}\\
		a_{5,5}&a_{3,6}'\\
		a_{6,5}&a_{2,6}'\\
		a_{3,5}'&a_{1,6}'\\
	a_{6,5}&a_{8,6}\\
	a_{5,5}'&a_{8,6}'\\
		a_{4,5}'&a_{2,6}'\\
		a_{6,5}'&a_{4,6}'
		\ema,
			\bma
		a_{1,5}&a_{1,7}\\
		a_{1,5}'&a_{2,7}\\	
		a_{3,5}&a_{3,7}\\
		a_{4,5}&a_{3,7}'\\
		a_{5,5}&a_{5,7}\\
		a_{6,5}&a_{5,7}'\\
		a_{3,5}'&a_{7,7}\\
		a_{6,5}&a_{1,7}'\\
		a_{5,5}'&a_{2,7}'\\
		a_{4,5}'&a_{5,7}'\\
		a_{6,5}'&a_{7,7}'
		\ema,
		\bma
	a_{1,6}&a_{1,7}\\
	a_{2,6}&a_{2,7}\\	
	a_{3,6}&a_{3,7}\\
	a_{4,6}&a_{3,7}'\\
	a_{3,6}'&a_{5,7}\\
	a_{2,6}'&a_{5,7}'\\
	a_{1,6}'&a_{7,7}\\
	a_{8,6}&a_{1,7}'\\
	a_{8,6}'&a_{2,7}'\\
	a_{2,6}'&a_{5,7}'\\
	a_{4,6}'&a_{7,7}'	
		\ema.
	\end{eqnarray}
By permutating the rows and columns of above matrices, One can verify that any $4\times2$ submatrices of them have at least three orthogonal pairs. It violates the fact that the $4\times2$ matrix has four orthogonal pairs.

(ii.e) When $|P|=5$,  we obtain that there exist five $\ket{\varphi_{m_1}},\ket{\varphi_{m_2}},\ket{\varphi_{m_3}}, \ket{\varphi_{m_4}}$ and $\ket{\varphi_{m_5}}$ in the set $P$ and remaining product vectors are in the set $Q$. Then we obtain that elements of the set $P$ are orthogonal to that of the set $Q$ on system $C_{\overline{\mathcal{S}}}$. From \eqref{eq:varphi}, we obtain that $\ket{b_{m_1},c_{m_1}},\ket{b_{m_2},c_{m_2}},\ket{b_{m_3},c_{m_3}},\ket{b_{m_4},c_{m_4}}$ and $\ket{b_{m_5},c_{m_5}}$ are not orthogonal to any $\ket{b_l,c_l}$, where $l\in\{1,\cdots,11\}\backslash\{m_1,m_2,m_3,m_4,m_5\}$.Denote $\{\ket{b_{m_1},c_{m_1}},\ket{b_{m_2},c_{m_2}},\ket{b_{m_3},c_{m_3}},\ket{b_{m_4},c_{m_4}}, \ket{b_{m_5},c_{m_5}} \}$ as a $5\times2$ matrix 
	\begin{eqnarray}
		\label{eq:ph}
		\bma
		b_{m_1}&c_{m_1}\\
		b_{m_2}&c_{m_2}\\
		b_{m_3}&c_{m_3}\\
		b_{m_4}&c_{m_4}\\
		b_{m_5}&c_{m_5}
		\ema.
	\end{eqnarray} Then they must satisfy one of the two cases (ii.e.i) and (ii.e.ii). 
	
	(ii.e.i) Suppose three of $\ket{b_{m_1}}, \ket{b_{m_2}},\ket{b_{m_3}},\ket{b_{m_4}}$ and $\ket{b_{m_5}}$ are same. We obtain that they are orthogonal to the remaining two elements. It implies that the remaining two elements are same. From \eqref{eq:11times7uom} and  \eqref{eq:p3}, we obtain that two of $\ket{c_{m_1}}, \ket{c_{m_2}},\ket{c_{m_3}},\ket{c_{m_4}}$ and $\ket{c_{m_5}}$ are same and there exist three orthogonal pairs between them. 
	Then we obtain that the $5\times2$ matrix in \eqref{eq:ph} has nine orthogonal pairs. From \eqref{eq:11times7uom}, the $5\times2$ matrix in \eqref{eq:ph} must be one of the four matrices
	\begin{eqnarray}
		\bma
		a_{1,3}&a_{1,4}\\
		a_{1,3}&a_{3,4}\\
		a_{1,3}&a_{3,4}'\\
		a_{1,3}'&a_{4,4}'\\
		a_{1,3}'&a_{7,4}'
		\ema,\quad
		\bma
		a_{1,3}&a_{1,5}\\
		a_{1,3}&a_{3,5}\\
		a_{1,3}&a_{6,5}\\
		a_{1,3}'&a_{5,5}'\\
		a_{1,3}'&a_{4,5}'
		\ema,\quad
		\bma
		a_{1,3}&a_{1,6}\\
		a_{1,3}&a_{3,6}\\
		a_{1,3}&a_{2,6}'\\
		a_{1,3}'&a_{8,6}'\\
		a_{1,3}'&a_{2,6}'
		\ema,\quad
		\bma
		a_{1,3}&a_{1,7}\\
		a_{1,3}&a_{3,7}\\
		a_{1,3}&a_{5,7}'\\
		a_{1,3}'&a_{2,7}'\\
		a_{1,3}'&a_{5,7}'
		\ema.
	\end{eqnarray}
	We obtain that the above $5\times2$ matrices at most have seven orthogonal pairs. It is a contradiction with the fact that $5\times2$ matrix in \eqref{eq:ph} has nine orthogonal pairs. So this case does not hold.
	Similarly, when three of $\ket{c_{m_1}}, \ket{c_{m_2}},\ket{c_{m_3}},\ket{c_{m_4}}$ and $\ket{c_{m_5}}$ are same, this case does not hold.
	
	(ii.e.ii) Suppose two of $\ket{b_{m_1}}, \ket{b_{m_2}},\ket{b_{m_3}},\ket{b_{m_4}}$ and $\ket{b_{m_5}}$ are same. Because $|P|=5$, from \eqref{eq:p3}, we obtain that they have three orthogonal pairs. If  two of $\ket{c_{m_1}}, \ket{c_{m_2}},\ket{c_{m_3}},\ket{c_{m_4}}$ and $\ket{c_{m_5}}$ are same, then we obtain that the $5\times2$ matrix in \eqref{eq:ph} has six orthogonal pairs. 
	
	\begin{eqnarray}
		\label{eq:2121}
		\bma
		a_{7,4}&a_{3,5}'\\
		a_{3,4}&a_{3,5}\\
		a_{3,4}'&a_{6,5}\\
		a_{3,4}'&a_{6,5}\\
		a_{1,4}'&a_{6,5}'
		\ema,\quad
		\bma
		a_{1,4}&a_{1,5}\\
		a_{3,4}&a_{3,5}\\
		a_{3,4}'&a_{6,5}\\
		a_{3,4}'&a_{6,5}\\
		a_{1,4}'&a_{6,5}'
		\ema,\quad
		\bma
		a_{2,4}&a_{2,6}\\
		a_{3,4}&a_{3,6}\\
		a_{3,4}'&a_{2,6}'\\
		a_{3,4}'&a_{8,6}\\
		a_{7,4}'&a_{2,6}'
		\ema,\quad
		\bma
		a_{2,4}'&a_{5,7}\\
		a_{3,4}&a_{3,7}\\
		a_{3,4}'&a_{5,7}'\\
		a_{3,4}'&a_{1,7}'\\
		a_{7,4}'&a_{5,7}'
		\ema,
		\nonumber\\
		\bma
		a_{1,5}'&a_{2,6}\\
		a_{6,5}&a_{2,6}'\\
		a_{6,5}&a_{8,6}\\
		a_{6,5}'&a_{4,6}'\\
		a_{4,5}'&a_{2,6}'
		\ema,\quad
		\bma
		a_{5,5}&a_{5,7}\\
		a_{6,5}&a_{5,7}'\\
		a_{6,5}&a_{1,7}'\\
		a_{6,5}'&a_{7,7}'\\
		a_{4,5}'&a_{5,7}'
		\ema,\quad
		\bma
		a_{2,6}&a_{2,7}\\
		a_{3,6}&a_{3,7}\\
		a_{3,6}'&a_{5,7}\\
		a_{2,6}'&a_{5,7}'\\
		a_{2,6}'&a_{5,7}'
		\ema,\quad
		\bma
		a_{2,6}&a_{2,7}\\
		a_{8,6}'&a_{2,7}'\\
		a_{3,6}'&a_{5,7}\\
		a_{2,6}'&a_{5,7}'\\
		a_{2,6}'&a_{5,7}'
		\ema.
	\end{eqnarray}
	Then by permutating the two columns of the above matrices, we obtain the remaining eight $5\times2$ matrices. It implies that the remaining matrices have same orthogonal pairs with the matrices in \eqref{eq:2121}. By investigating the number of the  orthogonal pairs in \eqref{eq:2121}, we obtain that they have at most five orthogonal pairs. It violate the fact that $5\times2$ matrix in \eqref{eq:ph} has six orthogonal pairs. Thus, we obtain that two cases in
	(ii.e) do not hold.%We show that the set $\Phi$ is nonlocally distinguishable by using the following five steps (i)-(v). %For the convenience of readers, we present steps (i)-(v) in Fig. \ref{fig:reducible1}. 

	Thus, we obtain that the set $\Phi$ is locally indistinguishable on systems $C_{\mathcal{S}}$ and $C_{\overline{{\mathcal{S}}}}$ when $\mathcal{S}=\{m,n\}$. 
\end{proof}

\section{The proof of Lemma \ref{le:eqodd}}
\label{le:lemma5}
Before we proof the  Lemma \ref{le:eqodd}, we introduce some notations and facts. %For investigating the number of orthogonal pairs of UOMs,  
Denote $\mathcal{O}(m,n)$ as the subset of the set consisting of all  $m\times n$ orthogonal matrices such that $\mathcal{O}(m,n)$ contains at least one UOM. Denote $A=[a_{i,j}]\in\mathcal{O}(m,n)$ and $A:=\bma A_1&A_2&\cdots&A_n\ema$, where $A_j$ is the $j$-th column of $A$. Denote $\s(A_j)$ as the number of independent vector variables in $A_j$ and $p_j$ the number of all orthogonal pairs in $A_j$. Recall the definition of $\textit{multiplicity}$ $\m(x,X)$. Suppose $X=[x_{i,j}]\in \mathcal{M}(m,n)$, where $\mathcal{M}(m,n)$ is the set of all $m\times n$ matrices. Denote $\m(x,X)$ of $x$ in $X$ as 
the number of pairs $(i, j)$ such that $x_{i,j}=x$. We shall simplify this notation $\m(x,X)$ by writing just $\m(x)$. Denote $A=[a_{i,j}]\in\mathcal{O}(m,n)$ as a UOM, then $a_{i,j}'$ occurs in $A$ for all $i,j$. In addition, we define $p_j$ in the above way as 
\begin{eqnarray}
	\label{eq:pj}
	p_j=\sum_{i=1}^{\s(A_j)}\m(a_{i,j})\m(a_{i,j}'),
\end{eqnarray} 
where the summation is over all pairs $\{a_{i,j},a_{i,j}'\}$ in $A_j$. Ref. \cite{Chen_2018} has shown the fact that 
\begin{eqnarray}
	\label{eq:inequalityp}
	\sum_jp_j\geq\frac{m(m-1)}{2}.
\end{eqnarray}
In the following, we show the proof of Lemma \ref{le:eqodd}.

\begin{proof}
	From Lemma \ref{le:functionp} (i) and the definition of $p_j$ in \eqref{eq:pj}, we obtain that for each column the minimum $p_j=1\times1+\cdots+1\times1=\frac{q+1}{2}$. So we obtain that the minimum number of orthogonal pairs of $q$-qubit UOMs with $q+1$ states is $\frac{q(q+1)}{2}$, where $q$ is odd. For example, it has been shown in \cite{dms03} that there exists a $(q+1)\times q$ UOM 
	$\bma
	0&0&0&\cdots&0&0&\cdots&0&0\\
	1&\psi_1&\psi_2&\cdots&\psi_{\frac{q-1}{2}}&\psi_{\frac{q-1}{2}}'&\cdots&\psi_2'&\psi_1'\\
	\psi_1'&1&\psi_1&\cdots&\psi_{\frac{q-1}{2}-1}&\psi_{\frac{q-1}{2}}&\cdots&\psi_3'&\psi_2'\\
	\vdots&\vdots&\vdots&\vdots&\vdots&\vdots&\vdots&\vdots&\vdots\\
	\psi_1&\psi_2&\psi_3&\cdots&\psi_{\frac{q-1}{2}}'&\psi_{\frac{q-1}{2}-1}'&\cdots&\psi_1'&1
	\ema$ such that each column has exactly $\frac{q+1}{2}$ orthogonal pairs. Then the sum of orthogonal pairs of the $(q+1)\times q$ UOM is $\frac{(q+1)q}{2}$. %So it corresponds to a complete graphs with four vertices. 
	So every $(q+1)\times q$ UOM with minimum size corresponds to a complete graph with $q+1$ vertices. We have proven the first claim of this lemma.
	
	Any two rows have exactly one orthogonal pair in all $(q+1)\times q$ UOMs. It implies that one vertex only appears in one edge. So the graphs of all columns of $(q+1)\times q$ UOMs are {isomorphic}. Thus $(q+1)\times q$ UOMs can be constructed with exactly one sort of {nonisomorphic} graph, where $q$ is odd. Hence we have proven the last claim of this lemma.
\end{proof}

\section{The proof of Theorem \ref{le:eqsort}}
\label{eq:pf}
\begin{proof}
	We show $6\times4$, $8\times6$ and $11\times8$ UOMs in case $(a), (b)$ and $(c)$, respectively. 
	$(a)$ We consider $6\times4$ UOMs. It has been proven that there exists only one $6\times4$ UOM in \cite{Johnston_2014}. We denote $B=[b_{i,j}]$ as the $6\times4$ UOM, $1\leq i\leq6, 1\leq j\leq4$. Then we may assume that $B_j$ is the $j$ column of $B$. So we have \begin{eqnarray}B=\bma B_1&B_2&B_3&B_4\ema= 
		\label{uom:64}
		\bma
		b_{1,1}&b_{1,2}&b_{1,3}&b_{1,4}\\
		b_{1,1}&b_{2,2}&b_{1,3}'&b_{2,4}\\
		b_{1,1}'&b_{1,2}&b_{3,3}&b_{3,4}\\
		b_{1,1}'&b_{2,2}&b_{4,3}&b_{3,4}'\\
		b_{5,1}&b_{2,2}'&b_{3,3}'&b_{1,4}'\\
		b_{5,1}'&b_{1,2}'&b_{4,3}'&b_{2,4}'
		\ema,
	\end{eqnarray}
	where $b_{5,1}\neq b_{1,1}, b_{1,1}', b_{2,2}\neq b_{1,2},b_{1,2}', b_{1,3}\neq b_{3,3}, b_{3,3}'\neq b_{4,3}, b_{4,3}'$. Because any two rows of $B$ have exactly one orthogonal pair, we obtain that $B$ corresponds to a complete graph with six vertices. We denote $V_i$ as the vertex corresponding to the $i$-th row of $B$, $1\leq i\leq6$. Then we present the complete graph with six vertices in Fig. \ref{fig:uom64} and the graphs corresponding to cases of orthogonal pairs of column $1-4$ in Fig. \ref{fig:uom64de}.

	\begin{figure}[htb]
		\includegraphics[scale=0.6,angle=0]{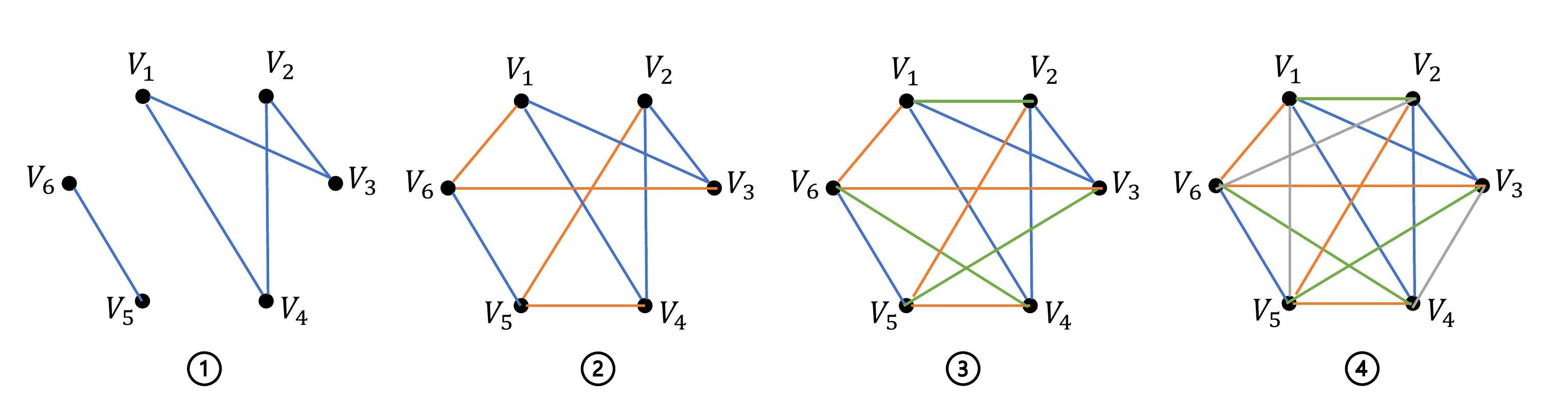}
		\caption{The vertex $V_i$ corresponds to the $i$-th row of $B$, $1\leq i\leq 6$. Graph \textcircled{1} - \textcircled{4}   imply the cases of orthogonal pairs in the first column 1 - 4, respectively. The blue edges in \textcircled{1} show the orthogonal pairs in the first column $B_1$ of $B$. Similarly, the yellow, green and grey edges imply the orthogonal pairs in $B_2, B_3, B_4$, respectively. The graph \textcircled{4} shows that $B$ is a complete graph with six vertices.
		} 
		\label{fig:uom64}
	\end{figure}
	
	\begin{figure}[htb]
		\includegraphics[scale=0.6,angle=0]{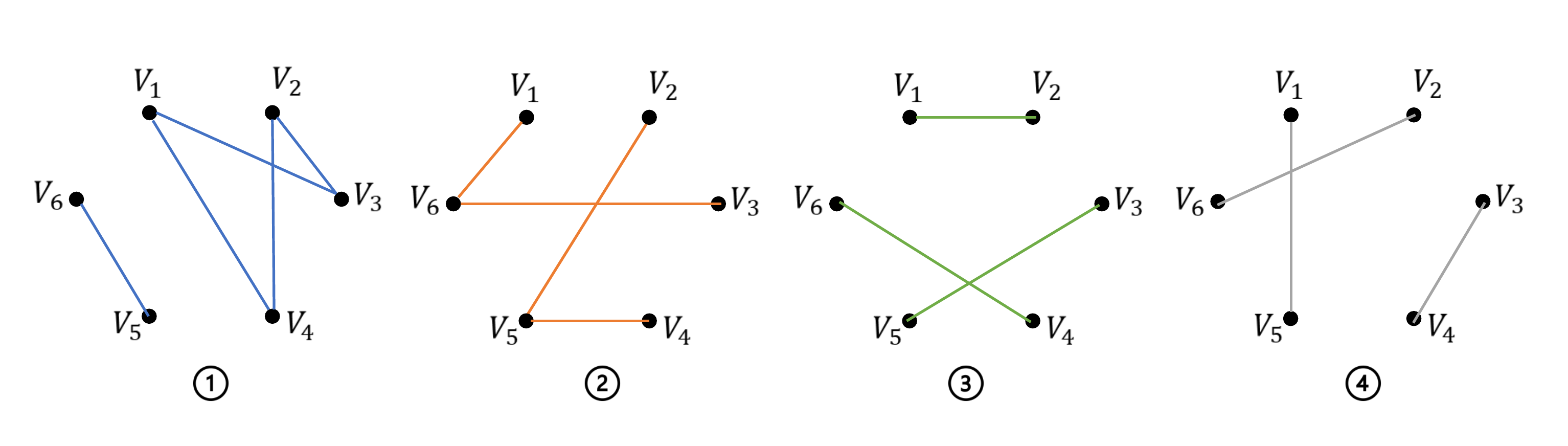}
		\caption{The vertex $V_i$ corresponds to the the $i$-th row of $B$. Graph \textcircled{1} - \textcircled{4} imply the cases of orthogonal pairs in column $1-4$, respectively. } 
		\label{fig:uom64de}
	\end{figure} 
	
	In Fig. \ref{fig:uom64de}, we obtain that \textcircled{3} and \textcircled{4} are two {isomorphic} graphs. So the $6\times4$ UOM can be constructed by three sorts of nonisomorphic graphs. 
	
	%\subsection{An $8\times6$ UOM}
	%\label{sub86}
	$(b)$ Denote $C=[c_{i,j}]$ as an $8\times6$ UOM. From \cite{Johnston_2014}, we rewrite it as  
	\begin{eqnarray}
		\label{uom:86}
		C=	\bma
		c_{1,1}&c_{1,2}&c_{1,3}&c_{1,4}&c_{1,5}&c_{1,6}\\
		c_{1,1}&c_{2,2}&c_{2,3}&c_{1,4}'&c_{2,5}&c_{2,6}\\
		c_{1,1}'&c_{1,2}&c_{2,3}&c_{3,4}&c_{3,5}&c_{3,6}\\
		c_{1,1}'&c_{2,2}&c_{1,3}&c_{3,4}'&c_{4,5}&c_{4,6}\\
		c_{5,1}&c_{5,2}&c_{2,3}'&c_{5,4}&c_{1,5}'&c_{4,6}'\\
		c_{5,1}'&c_{2,2}'&c_{6,3}&c_{6,4}&c_{3,5}'&c_{1,6}'\\
		c_{7,1}&c_{1,2}'&c_{6,3}'&c_{5,4}'&c_{4,5}'&c_{2,6}'\\
		c_{7,1}'&c_{5,2}'&c_{1,3}'&c_{6,4}'&c_{2,5}'&c_{3,6}'\\
		\ema.
	\end{eqnarray}
	
	From \eqref{uom:86}, we have $\bma p_1&p_2&p_3&p_4&p_5&p_6\ema=\bma 6&5&5&4&4&4\ema$. Then we have $\sum_{j=1}^6 p_j=6+5+5+4+4+4=28$. So we obtain that any two rows have exactly one orthogonal pair. $C$ corresponds to a complete graph with eight vertices. denote $V_i$ as the vertex corresponding to the $i$-th row, $1\leq i\leq8$. Then we present the complete graph with six vertices in Fig. \ref{fig:uom86} and the graphs corresponding to cases of orthogonal pairs of column $1-6$ in Fig. \ref{fig:uom86de}. 
	
	\begin{figure}[htb]
		\includegraphics[scale=0.5,angle=0]{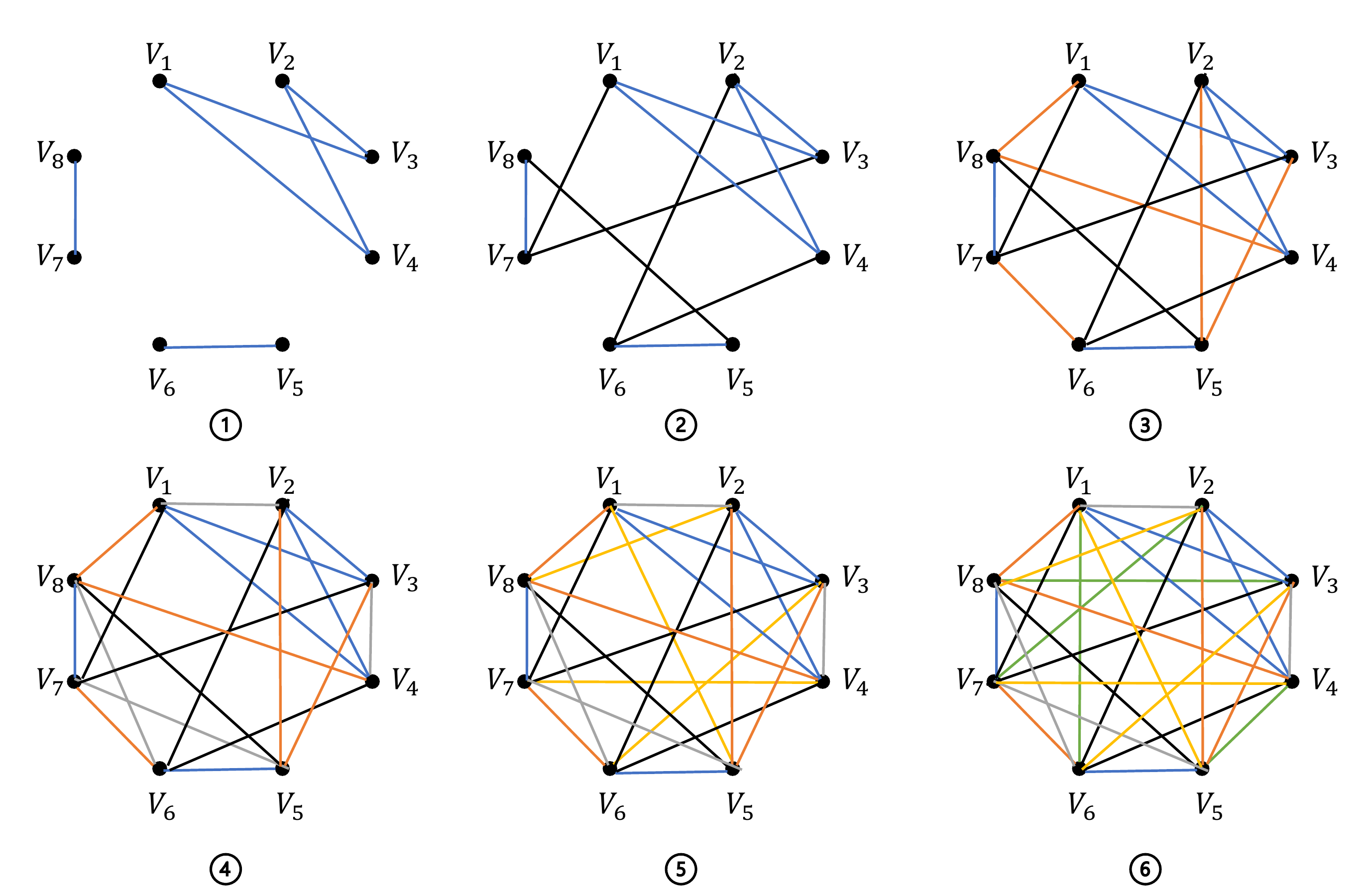}
		\caption{The vertex $V_i$ corresponds to the $i$-th row of $C$, $i=1,\cdots,6$. Graph $\textcircled{1}-\textcircled{6}$ correspond to the cases of orthogonal pairs of the first column $1-6$. Graph \textcircled{6} shows the complete graph corresponding to $C$. } 
		\label{fig:uom86}
	\end{figure}
	
	\begin{figure}[htb]
		\includegraphics[scale=0.5,angle=0]{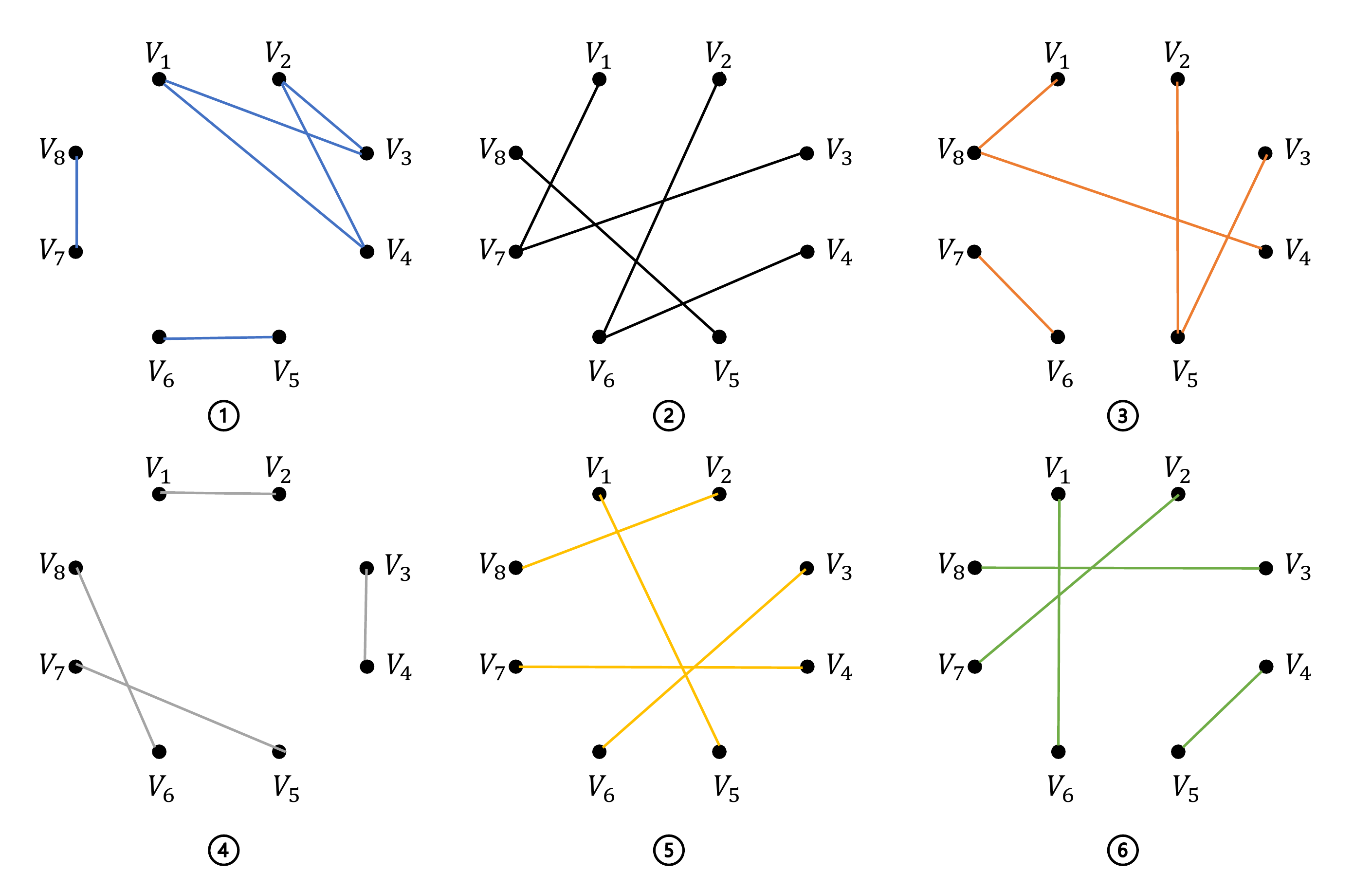}
		\caption{Graph \textcircled{1} - \textcircled{6} correspond to the cases of orthogonal pairs in column $1-6$ of $C$, respectively.} 
		\label{fig:uom86de}
	\end{figure}
	
	In Fig. \ref{fig:uom86de}, we obtain that \textcircled{2}, \textcircled{3} are isomorphic and \textcircled{4}, \textcircled{5}, \textcircled{6} are isomorphic. So the $8\times6$ UOM can be constructed by three sorts of nonisomorphic graphs. 

	$(c)$ It has been shown in \cite{2013The} that the $11\times8$ UOM corresponds a complete graph. Denote $D=[d_{i,j}]$ as the $11\times8$ UOM, then we rewrite it as 
	\begin{eqnarray}
		\label{eq:mat118}
		D=
		\bma
		d_{1,1}&d_{1,2}&d_{1,3}&d_{1,4}&	d_{1,5}&d_{1,6}&d_{1,7}&d_{1,8}\\
		d_{1,1}'&d_{2,2}&d_{2,3}&d_{2,4}&	d_{2,5}&d_{2,6}&d_{2,7}&d_{2,8}\\
		d_{3,1}&d_{2,2}'&d_{3,3}&d_{3,4}&	d_{1,5}'&d_{3,6}&d_{3,7}&d_{3,8}\\
		d_{4,1}&d_{3,2}&d_{4,3}&d_{3,4}'&	d_{3,5}&d_{1,6}'&d_{4,7}&d_{2,8}'\\
		d_{5,1}&d_{1,2}'&d_{2,3}'&d_{4,4}&	d_{4,5}&d_{3,6}'&d_{4,7}'&d_{4,8}\\
		d_{3,1}'&d_{3,2}'&d_{1,3}&d_{4,4}'&	d_{5,5}&d_{2,6}'&d_{1,7}'&d_{5,8}\\
		d_{5,1}'&d_{7,2}&d_{4,3}'&d_{2,4}'&	d_{5,5}'&d_{4,6}&d_{3,7}'&d_{1,8}'\\
		d_{4,1}'&d_{7,2}&d_{3,3}'&d_{1,4}'&	d_{4,5}'&d_{4,6}'&d_{2,7}'&d_{5,8}'\\
		d_{4,1}'&d_{7,2}'&d_{1,3}'&d_{4,4}'&	d_{2,5}'&d_{5,6}&d_{5,7}&d_{3,8}'\\
		d_{5,1}'&d_{7,2}'&d_{3,3}'&d_{1,4}'&	d_{3,5}'&d_{5,6}'&d_{2,7}'&d_{5,8}'\\
		d_{3,1}'&d_{3,2}'&d_{1,3}'&d_{1,4}&	d_{5,5}&d_{2,6}'&d_{5,7}'&d_{4,8}'\\
		\ema.
	\end{eqnarray}
	
	From \eqref{eq:mat118}, we have $\bma p_1&p_2&p_3&p_4&p_5&p_6&p_7&p_8\ema=\bma7&8&8&8&6&6&6&6\ema$. Then we have $\sum_jp_j=55$. We obtain that any two rows has exactly one orthogonal pair. So $D$ corresponds to a complete graph with eleven vertices. We present the complete graph and the graphs corresponding to the cases of orthogonal pairs of column $1-8$ in Fig. \ref{fig:uom118} and Fig. \ref{fig:uom118de}, respectively. 
	\begin{figure}[htb]
		\includegraphics[scale=0.5,angle=0]{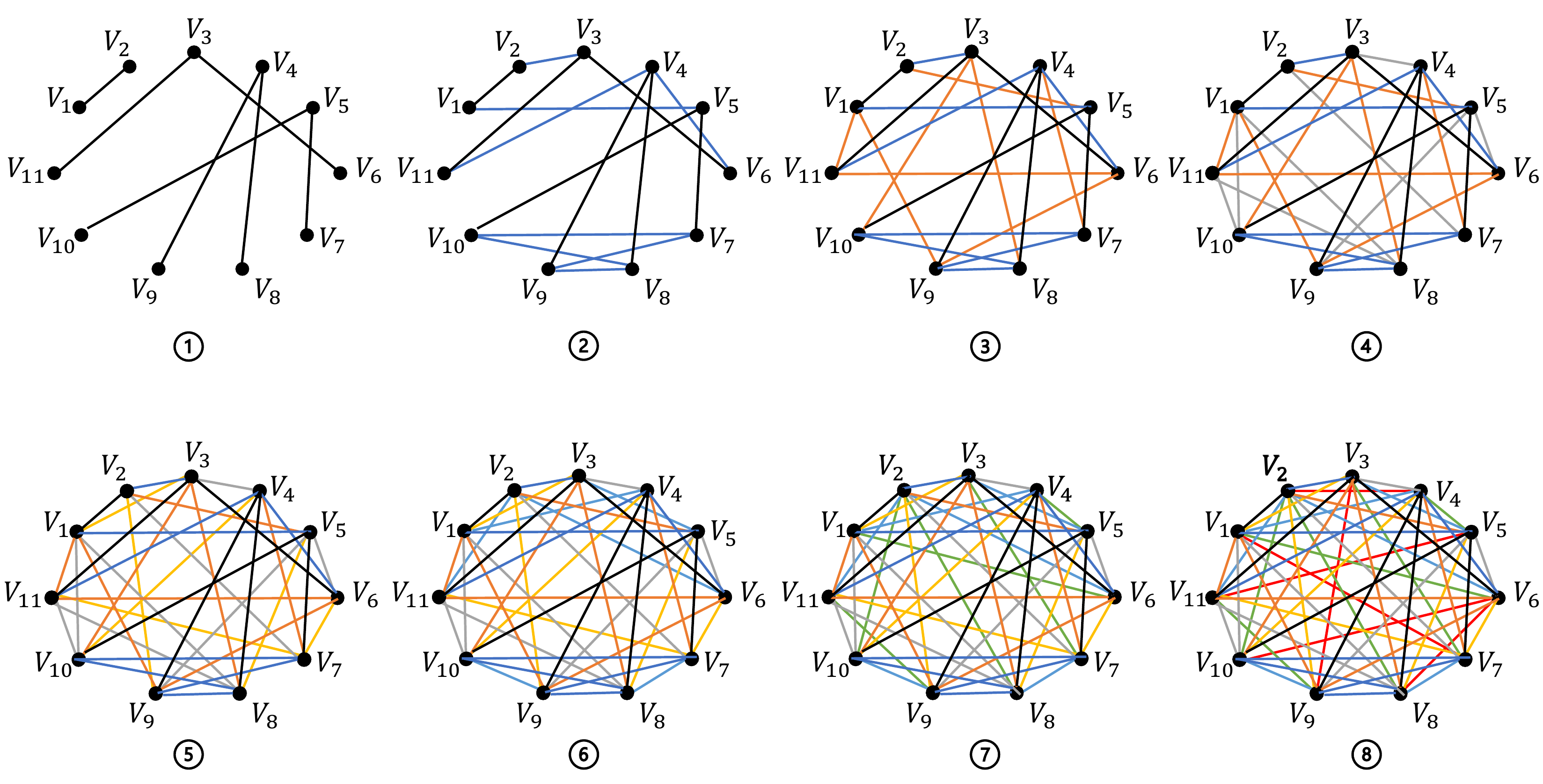}
		\caption{The vertex $V_i$ corresponds to the $i$-th row of $D$. Graph $\textcircled{1}-\textcircled{8}$ correspond to the cases of orthogonal pairs of the first column $1-8$.} 
		\label{fig:uom118}
	\end{figure}
	
	\begin{figure}[htb]
		\includegraphics[scale=0.5,angle=0]{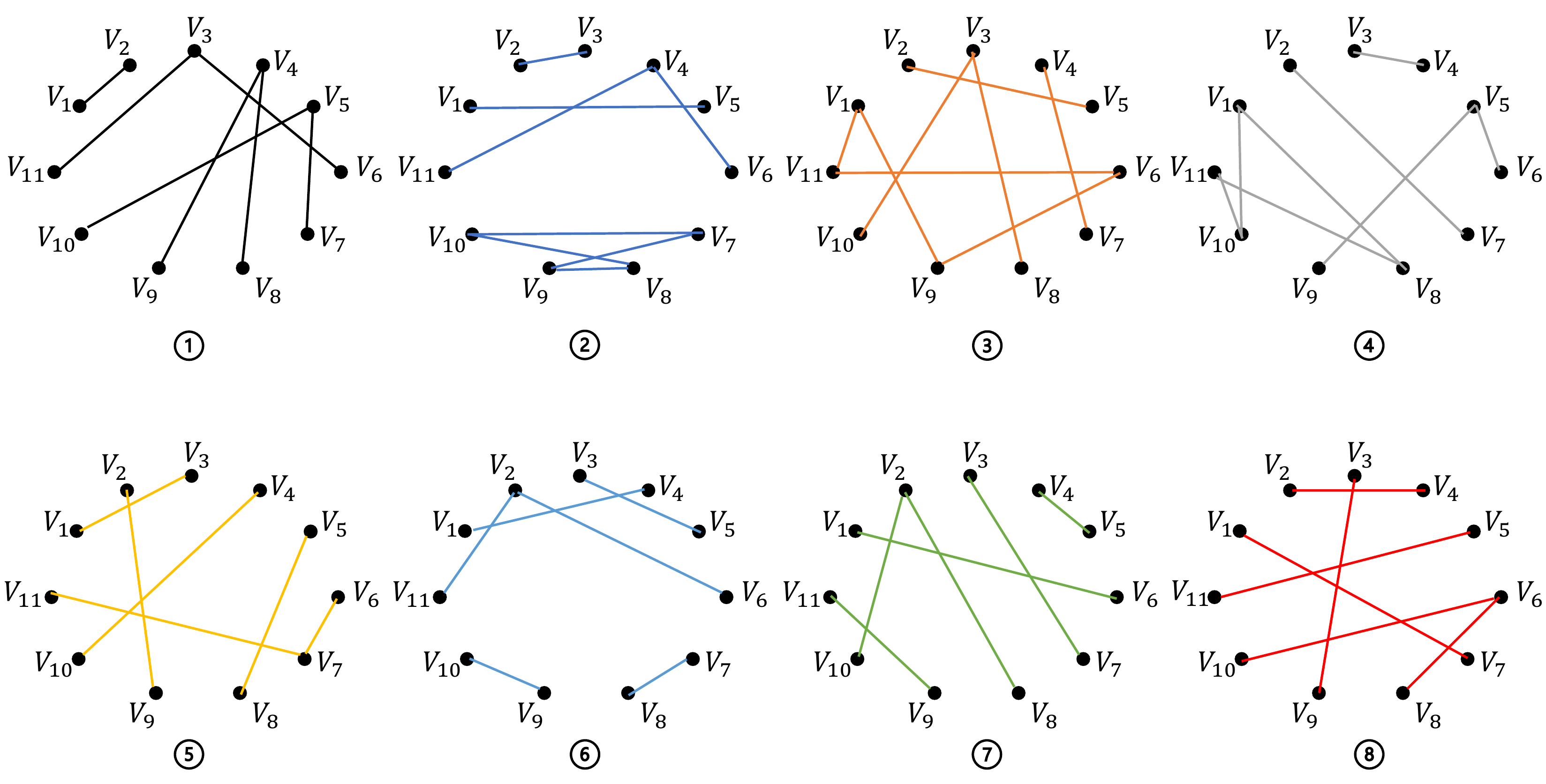}
		\caption{ Graph $\textcircled{1}-\textcircled{8}$ imply the cases of orthogonal pairs in column $1-8$ of $D$, respectively.
		} 
		\label{fig:uom118de}
	\end{figure}
	
	In Fig. \ref{fig:uom118de}, we obtain that \textcircled{2}, \textcircled{3}, \textcircled{4} are isomorphic and \textcircled{5}, \textcircled{6}, \textcircled{7}, \textcircled{8} are isomorphic. 
	%So $D$ can be constructed by three sorts of nonisomorphic graphs. 
	So there exists an $11\times8$ UOM such that it can be constructed by three sorts of nonisomorphic graphs. 
	
	%In graph theory, suppose $G(V,E)$ and $G_1(V_1,E_1)$ are two graphs, where $V,V_1$ are two sets of vertices and $E,E_1$ are sets of edges in graph $G, G_1$, respectively. If there exists a bijective $M:V\rightarrow V_1$ such that for all $X,Y\in V$, $XY\in E$ is equivalent to $m(X)m(Y)\in E_1$, then $G$ and $G_1$ are {isomorphic}. Then we consider the number of {nonisomorphic} graphs for constructing some UOMs. 
	
	%From Lemma \ref{le:functionp} (i) and the definition of $p_j$ in \eqref{eq:pj}, we obtain that the minimum $p_j=1\times1+\cdots+1\times1=\frac{q+1}{2}$ for each column. So we obtain that the minimum number of orthogonal pairs of $q$-qubit UOMs with $q+1$ states is $\frac{q(q+1)}{2}$, where $q$ is odd. So every $q$-qubit UOM with minimum size corresponds to a complete graph with $q$ vertices. Then there is exactly one orthogonal pair in all $q+1$ $q$-qubit UOMs. It implies that one vertex only appears in one edge. So the graphs of all columns of $q+1$ $q$-qubit UOMs are {isomorphic}. Thus $q+1$ $q$-qubit UOMs can be constructed with one {nonisomorphic} graph, where $q$ is odd. 
	
	Therefore, %we obtain that $(q+1)\times q$ UOMs can be constructed by one sort of nonisomorphic graph, where $q$ is odd. 
	from case $(a), (b)$ and $(c)$% \ref{sub64}, \ref{sub86} and  \ref{sub118}
	, we obtain that the number of sorts of  {nonisomorphic} graphs for constructing $6\times4$, $8\times6$ and $11\times8$ UOMs in \eqref{uom:64}, \eqref{uom:86}, \eqref{eq:mat118} are three.
\end{proof}
\bibliographystyle{unsrt}

\bibliography{20200615upb}

\end{document}